\documentclass[review]{elsarticle}

\usepackage{lineno,hyperref}
\usepackage{multirow}
\usepackage{graphicx}
\usepackage{enumitem}

\modulolinenumbers[5]

\journal{Journal of \LaTeX\ Templates}

\usepackage{booktabs}
\usepackage{amsthm}
\usepackage{amsmath}

\renewcommand{\paragraph}[1]{\vspace{1mm}\noindent{\bf #1}}
\newcommand{\eat}[1]{}

\newtheorem{definition}{Definition}[section]
\newtheorem{proposition}[definition]{Proposition}
\newtheorem{lemma}[definition]{Lemma}

\newtheorem{corollary}[definition]{Corollary}

\newtheorem{theorem}[definition]{Theorem}

\newtheorem{example}[definition]{Example}

\newcommand{\symif}{:-}
\newcommand{\imply}{\Rightarrow}

\newcommand{\squishlist}{
	\begin{list}{$\bullet$}
		{
			\setlength{\itemsep}{0pt}
			\setlength{\parsep}{3pt}
			\setlength{\topsep}{3pt}
			\setlength{\partopsep}{0pt}
			\setlength{\leftmargin}{1.5em}
			\setlength{\labelwidth}{1em}
			\setlength{\labelsep}{0.5em} } }

\newcommand{\CM}{{\cal M}}
\newcommand{\C}{{\cal C}}
\newcommand{\V}{{\cal V}}
\newcommand{\R}{{\cal R}}
\newcommand{\I}{{\cal I}}
\newcommand{\CP}{{\cal P}}
\newcommand{\CS}{{\cal S}}

\newcommand{\CL}{{\cal L}}
\newcommand{\CQ}{{\cal Q}}
\newcommand{\CA}{{\cal A}}
\newcommand{\CB}{{\cal B}}
\newcommand{\CF}{{\cal F}}
\begin{document}

\begin{frontmatter}

\title{On the Complexity of Query Containment and Computing Certain Answers in the Presence of ACs}

\author{Foto N. Afrati}
\ead{afrati@gmail.com}
\author{Matthew Damigos}
\ead{mgdamig@gmail.com}

\address{National Technical University of Athens}

\begin{abstract}

We often add arithmetic to extend the expressiveness of query languages and study the complexity of problems such as
testing query containment and finding certain answers in the framework of answering queries using views.
When adding arithmetic comparisons, the complexity of such problems is higher than the complexity of their counterparts without them.
It has been  observed that we can achieve lower complexity if we restrict some of the comparisons in the containing query to be closed or open semi-interval comparisons. Here, focusing  a) on the problem of
containment for conjunctive queries  with arithmetic comparisons (CQAC queries, for short), we prove
upper bounds on its computational complexity and b)  on the problem of computing certain answers, we find large classes of CQAC queries and views where this problem is polynomial.

\end{abstract}


\end{frontmatter}


\section{Introduction}
\label{intro-sect}
%
For conjunctive queries, the query containment problem is NP-complete~\cite{Chandra77-1}.
When we have constants that are numbers (e.g., they may represent prices, dates, weights, lengths, heights) then, often, we want to compare them by checking, e.g., whether two numbers are equal or whether one is greater than the other, etc.
To  reason about numbers we want to have a more expressive language than conjunctive queries and, thus, we add arithmetic comparisons to the definition of the query.
We know that the query containment problem for conjunctive queries with arithmetic comparisons is $\Pi^p_2$-complete~\cite{Klug88-1,vdmeyden92,Kolaitis98-1}.
In previous literature \cite{WangTM05,Afrati04-1,AfratiLM06,Afrati06}, it has been noticed that there are classes of CQACs for which the query containment problem remains in NP and these classes can be syntactically characterized. In this paper we find much broader such classes of queries.

Query containment and many problems in answering queries using views are closely related.
In particular, in the framework of answering queries using views, we want to find all certain answers of the query on a given view instance, i.e., all the answers that are provable ``correct.''
A popular way for answering queries using views is by finding rewritings of the query in terms of the views that are contained in the query. There may exist many contained rewriting in a certain query language. We want to find the maximal contained rewriting (MCR for short) that contains all the rewritings, if there exists such a rewriting. 

The main results in this paper are the following:

{\bf Query Containment} We solve an open problem mentioned in  \cite{Afrati19} by extending significantly the class of CQAC queries that admit an NP containment test.
As concerns closed arithmetic comparisons, we think we are close to the boundary between the problem being in NP and being in $\Pi^p_2$.
The class of queries we consider includes (but is broader than)  the following case: The contained query is allowed to have any closed arithmetic comparisons.  The containing query is allowed to have any closed arithmetic comparisons that involve the head variables, but not between a head variable and a body variable. Moreover, the comparisons that are allowed in the body variables are the following: Several left semi-interval arithmetic comparisons and at most one right semi-interval arithmetic comparison.

This result is proven via a transformation of the queries to a Datalog query (for the containing query) and a conjunctive query (for the contained query) and reducing checking containment between these two. This result captures all results in \cite{Afrati19} but in a new way that allows us to further use the transformation to compute MCRs and certain answers in the framework of the problem of answering queries using views.

{\bf MCRs} We extend the results in \cite{Afrati04-1} and prove that we can find an MCR in the language of Datalog with arithmetic comparisons in the case where the query has the restrictions of the containing query above and the views use any closed ACs, except ACs between the head and non-head variables. In \cite{Afrati04-1}, only semi-interval arithmetic comparisons were allowed in the query and in the views.

{\bf Computing certain answers} We show for the first time how to compute certain answers in polynomial time using MCRs for the case the conjunctive queries have arithmetic comparisons. 

\vspace*{.2cm}
In proving the above main results, we needed to prove intermediate results, which could be extended beyond what was necessary for proving the main results. Those intermediate results are the 
following:

a) We defined the class of semi-monadic Datalog queries which include the class of monadic queries. We proved that checking containment of a conjunctive query to a 
semi-monadic Datalog query is NP-complete. 

b) We proved that if there is an MCR in the language of (possibly infinite) conjunctive queries with arithmetic comparisons, then, even in the presence of dependencies, we can use this MCR to compute all certain answers for any CQAC query and any CQAC views. 

The structure of the paper is as follows:

Section 5 proves the result about query containment.  The proofs in this section need technical results about the implication problem of arithmetic comparisons that are presented in the early subsections of Section 5.  They also need a   result about query containment of a conjunctive query to a semi-monadic Datalog query. This  is presented in Appendix C. Moreover, in Appendix B, we present the proof of the main technical result of Section 5.


Sections~\ref{sec-6} and~\ref{subsec-mcr-datalogAC} consider the framework of answering queries using views. 
%
Section~\ref{subsec-mcr-datalogAC} uses the findings of Section 5 to compute an MCR using the language Datalog$^{AC}$. 
Section~\ref{sec-6} shows the relation between the output after computing an MCR on a view instance and the output after computing the certain answers on the same instance. 
Appendix E starts a discussion on a new direction concerning unclean data and MCRs.
We include Appendices A and D, for reasons of completeness.

\subsection{Related work } 
\textit{CQ and CQAC containment:} The problem of containment between conjunctive queries (CQs, for short) has been studied in \cite{Chandra77-1}, where the authors show that the problem is NP-complete, and the containment can be tested by finding a containment mapping. As we already mentioned, considering CQs with arithmetic comparisons (CQACs), the problem of query containment is $\Pi^p_2$-complete \cite{vdmeyden92}. 
Zhang and Ozsoyoglu, in \cite{Zhang93-1}, showed that the testing containment of two CQACs can be done by checking the containment entailment.
Kolaitis et al. \cite{Kolaitis98-1} studied the computational complexity of the query-containment problem of queries with disequations ($\neq$). In particular, the authors showed that the problem remains $\Pi^p_2$-hard even in the cases where the acyclicity property holds and each predicate occurs at most three times. However, they proved that if each predicate occurs at most twice then the problem is in coNP. 
Karvounarakis and Tannen, in \cite{karvounarakis2008conjunctive}, also studied CQs with disequations ($\neq$) and identified special cases where query containment can be tested through checking for a containment mapping (i.e., the containment problem for these cases is NP-complete).

The homomorphism property for query containment of conjunctive queries with arithmetic comparisons was  studied in \cite{Klug88-1,Zhang94-4,Afrati04-1,Afrati19}. 
In \cite{Afrati04-1}, Afrati et al. investigated cases where the normalization step is not needed. They also identified classes of CQACs where the homomorphism property holds.
In \cite{Afrati19}, Afrati showed that the problem of containment of  two CQACs such that the homomorphism property holds is in NP. This work also identifies certain classes of CQACs where the homomorphism property holds. 
The containment of certain subclasses of CQACs that the homomorpshim property holds are also identified in \cite{Zhang94-4}.

\vspace{10px}
\textit{Datalog containment:} Since one of our results concerns monadic Datalog queries and the containment problem we briefly list some related results.
Although to test the containment of two Datalog queries is undecidable \cite{Shmueli93-1}, the containment of  a CQ in a Datalog query is decidable.
In the general case of non-monadic Datalog query, the problem of containment of a CQ in a recursive Datalog query is EXPTIME-complete \cite{Cosmadakis86-1, Chandra81-1, Sagiv88-1}. As \cite{Cosmadakis88-1} shows, the containment between monadic Datalog queries is decidable. In \cite{Chaudhuri92-1}, the containment problem of a Datalog query in a conjunctive query is proven to be doubly exponential.
The containment of a CQ in a linear monadic Datalog (i.e., each rule has at most one IDB) 
is NP-complete~\cite{Chaudhuri94-1}. In this work, we extend this result for any monadic Datalog query (exdending it also to a wider class, called semi-monadic Datalog queries). Recent work on containment problem for monadic Datalog  includes \cite{BenediktBS12}. 

\vspace{10px}
\textit{Rewritings and finding MCRs:} The problem of answering queries using views has been extensively investigated in the relevant literature (e.g., \cite{Levy95-1, Levy00-1, 2017Afrati}); including finding equivalent and contained rewriting. Algorithms for finding maximally contained rewritings (MCRs) have also been studied in the past \cite{Abiteboul98-1, Grahne99-1, Levy96-1, Pottinger00-1, Mitra99-1, Duschka97-3, Afrati99-1, AfratiLM06}.  The authors in \cite{Pottinger00-1} and \cite{Mitra99-1} propose two algorithms, the Minicon and shared variable algorithm, respectively, for finding MCRs in the language of unions of CQs when both queries and views are CQs.  \cite{Pottinger00-1} also considers restricted cases of arithmetic comparisons (LSI and RSIs) in both queries and views. In \cite{Abiteboul98-1}, the queries have inequalities ($\neq$), while the views are CQs. As it is also proven in this work, the data complexity of finding certain answers is co-NP hard. The works in \cite{Duschka97-3} and \cite{Afrati99-1} studied the problem where the query is given by a Datalog query, while the views are given by CQs and union of CQs, respectively. In both papers, the language of MCRs is  Datalog. The authors in \cite{CaoFGL18} studied the problem of finding MCRs in the framework of bounded query rewriting. They investigated several query classes, such as CQs, union of CQs, and first order queries, and analyzed the complexity in each class. Afrati and Kiourtis in \cite{AfratiK10} proposed an efficient algorithm that finds MCRs in the language of union of CQs in the presence of dependencies. The work in \cite{AfratiLM06} investigated the problem of finding MCRs for special cases of CQACs.

\vspace{10px}
\textit{Certain answers and MCRs:} The problem of finding certain answers has been extensively investigated in the context of data integration and data exchange, the last 20 years (e.g., \cite{AndritsosFFHHHKMNPVVY02, Abiteboul98-1, Grahne99-1, fagin2005data, AfratiK10, KonstantinidisA13}). In \cite{Grahne99-1, fagin2005data}, the authors investigated the problem of finding certain answers in the context of data exchange, considering CQs. The work in \cite{fagin2005data} was extended for arithmetic and linear arithmetic CQs  in \cite{CateKO13}, where the authors proved that the problem of finding certain answers into target schema is co$\exists$R-complete in data complexity. 
In \cite{Abiteboul98-1}, the authors investigated the relationship between MCRs and certain answers.
In \cite{AfratiK10}, the authors proved that an MCR of a union CQs computes all the certain answers, where MCR is considered in the language of union of CQs. 
In many of the works about finding certain answers the chase algorithm is used as a tool. 
Recent studies of the  chase in the framework of query containment and data integration include \cite{BenediktKMMPST17} and \cite{KonstantinidisA14}.

\vspace{10px}
\textit{Other work with arithmetic comparisons in queries:} As concerns studying other related problems in the presence of arithmetic comparisons recent work can be found
in \cite{FanLLT18}, where the authors propose to extend
graph functional dependencies with linear arithmetic expressions
and arithmetic comparisons. They study the problems of testing satisfiability and related problems over integers (i.e., for non-dense orders).
A thorough study of the complexity of the problem of evaluating conjunctive queries  with inequalities ($\neq$) is done in \cite{KoutrisMRS15}. In  \cite{PapadimitriouY97} the complexity  of evaluating conjunctive queries  with
arithmetic comparisons is investigated  for acyclic queries, while query containment for acyclic conjunctive queries was investigated in \cite{ChekuriRaj97}.
Recent works \cite{CateKO13,AfratiLP08}  have added arithmetic to extend the expressiveness of  tuple generating dependencies and data exchange mappings, and studied the complexity of related problems. Queries with arithmetic comparisons on incomplete databases are considered in \cite{ConsoleHL20}.

\section{Preliminaries}


A \textit{relation schema} is a named relation defined by its name (called \textit{relation name} or \textit{relational symbol}) and a vector of attributes. An \textit{instance} of a relation schema is a collection of tuples with values over its attribute set. These tuples are called {\em facts}. The schemas of the relations in a database constitute its \textit{database schema}. A relational \textit{database instance} (database, for short) is a collection of stored relation instances.

A {\em conjunctive query (CQ in short)} $Q$ over a database schema $\CS$ is a query of
the form: $h( \overline{X})\ :-\  e_1( \overline{X}_1),\ldots, e_k( \overline{X}_k)$,
where $h( \overline{X})$ and $ e_i( \overline{X}_i)$ are atoms, i.e., 
they contain a relational symbol (also called \textit{predicate} - here, $h$ and $e_i$ are predicates) and a
vector of variables and constants.
The atoms that contain only constants are called \textit{ground} atoms and they represent   \textit{facts}.

The {\em head} $h(\overline{X})$, denoted $head(Q)$,
represents the results of the query, and $e_1 \ldots e_k$ represent
database relations (also called base relations) in $\CS$.
The variables in $\overline{X}$
are called \textit{head} or \textit{distinguished} variables,
while the variables in $\overline{X}_i$ are called \textit{body} or \textit{nondistinguished}
variables of the query. 
The part of the conjunctive query on the right of symbol $:-$ is called the {\em body} of the query and is denoted $body(Q)$.
Each atom in the body of a conjunctive query
is said to be a {\em subgoal}. 
A conjunctive query is said to be
\textit{safe} if all its distinguished variables also occur in its
body. We only consider safe queries here.

The \textit{result} (or \textit{answer}), denoted $Q(D)$, of a CQ $Q$ when it is applied on a database instance $D$ is the set of atoms such that for each assignment $h$ of variables of $Q$ that makes all the atoms in the body of $Q$ true the atom $h(head(Q))$ is in $Q(D)$. 

{\em Conjunctive queries with arithmetic comparisons (CQAC for short)} are conjunctive queries that, besides the
{\em ordinary} relational subgoals, use also builtin subgoals that are arithmetic comparisons (AC for short), i.e., of the form
$X\theta Y$ where $\theta$ is one of the following: $<, >, \leq, \geq, =, \neq$. Also, $X$ is a variable and
$Y$ is either a variable or constant. If $\theta$ is either $<$ or $>$ we say that it is an open arithmetic comparison
and if $\theta$ is either $\leq$ or $\geq$ we say that it is a closed AC. If the AC is either of the form $X<c$ or $X\leq c$ (respectively, either $X>c$ or $X\geq c$, resp.), where $X$ is a variable and $c$ is a constant,  then it is called \textit{left semi-interval}, LSI for short (\textit{right semi-interval}, RSI for short, respectively).
In the following, we use the notation $Q=Q_0+\beta$ to describe a CQAC query $Q$, where
$Q_0$ are the relational subgoals of $Q$ and $\beta$ are the arithmetic comparison
subgoals of $Q$. We define the {\em closure} of a set of ACs to be all the ACs that are implied by this set of ACs.
The result $Q(D)$ of a CQAC $Q$, when it is applied on a database $D$, is given by taking all the assignments of variables (in the same fashion as  CQs) such that the  atoms in the body  are included in $D$ and the ACs are true. For each assignment where these conditions are true, we produce a fact in the output  $Q(D)$.

All through this paper, we assume the following setting without mentioning it again.
\begin{enumerate}
	\item Values for the arguments in the arithmetic comparisons are
	chosen from an infinite, totally densely ordered set, such as the
	rationals or reals.
	\item The arithmetic comparisons are not contradictory (or, otherwisie, we say that they are consistent); that is,
	there exists an instantiation of the variables such that all the
	arithmetic comparisons are true.
	\item All the comparisons are safe, i.e., each variable in the
	comparisons also appears in some ordinary subgoal.
\end{enumerate}

A \textit{union of CQs} (resp. CQACs) is defined by a set $\CQ$ of CQs (resp. CQACs) whose heads have the same arity, and its answer $\CQ(D)$ is given by the union of the answers of the queries in $\CQ$ over the same database instance $D$; i.e., $\CQ(D)=\bigcup_{Q_i\in\CQ}Q_i(D)$.

A query $Q_1$ {\em is contained} in a query $Q_2$, denoted
$Q_1 \sqsubseteq Q_2$, if for any database $D$ of the base
relations, the answer  computed by $Q_1$ is a subset of the answer
computed by $Q_2$, i.e., $Q_1(D) \subseteq Q_2(D)$. The two
queries are {\em equivalent}, denoted $Q_1 \equiv Q_2$, if $Q_1
\sqsubseteq Q_2$ and $Q_2 \sqsubseteq Q_1$.

A {\em homomorphism} $h$ from a set of relational atoms $\CA$ to another set of relational atoms $\CB$ is a mapping
of variables and constants from one set to variables or constants of the other set that maps
each variable to a single variable or constant and each constant to the same constant.
Each atom of the former set should map to  an atom of the latter set with the same relational symbol. We also say that the homomorphism $h'$ from a set $\CA'\supseteq\CA$ is an  \textit{extension} of $h$ if for each variable or constant $x$ in $\CA'\cap\CA$ we have $h'(x)=h(x)$.

A {\em containment mapping} from a conjunctive query $Q_1$ to a conjunctive query $Q_2$  is a homomorphism from the atoms in the body of $Q_1$ to the atoms in the body of $Q_2$ that maps
the head of $Q_1$ to the head of $Q_2$.  All the mappings we refer to in this paper are containment mappings unless we say otherwise. Chandra and Merlin \cite{Chandra77-1} show that a conjunctive query $Q_1$
is contained in another conjunctive query $Q_2$ if and only if
there is a containment mapping from $Q_2$ to $Q_1$. The query containment problem for CQs is NP-complete.


\subsection{Testing query containment for CQACs}
\label{subsec:canonical-dbs}
In this section, we describe two   tests for CQAC query containment; using containment mappings and using canonical databases. 

In the rest of the paper, we denote with $Q_1=Q_{10}+\beta_1$ and
	$Q_2=Q_{20}+\beta_2$ the containing and contained query, respectively, where $Q_{10}$ denotes the relational atoms in the body of $Q_1$ and $\beta_1$ denotes the ACs. Similarly for query $Q_2$.

First, we present the test using containment mappings (see, e.g., in \cite{2017Afrati}). Although finding a single containment mapping suffices to test query containment for CQs (see the previous section), it is not enough in the case of CQACs. In fact, all the containment mappings from the containing query to the contained one should be considered. Before we describe how containment mappings can be used in order to test query containment between two CQACs, we define the concept of normalization of a CQAC.

\begin{definition}
	\label{dfn-normalization}
	Let $Q_1$ and $Q_2$ be two conjunctive queries with arithmetic
	comparisons (CQACs).  We want to test whether $Q_2
	\sqsubseteq  Q_1$. To do the testing, we  first  normalize
	each of $Q_1$ and $Q_2$ to $Q'_1$ and $Q'_2$, respectively. We {\em normalize} a CQAC query as
	follows:
	
	\begin{itemize}
		\item For each occurrence of a shared variable $X$ in a normal (i.e., relational) subgoal,
		except for the first occurrence, replace the occurrence of $X$ by a fresh
		variable $X_i$, and add $X = X_i$ to the comparisons of the
		query; and
		
		\item For each constant $c$ in a  normal subgoal,  replace the constant by a
		fresh variable $Z$, and add $Z = c$ to the comparisons of the
		query.
	\end{itemize}
\end{definition}

Theorem~\ref{thm:cont-CQAC}\cite{Gupta94-1,Zhang94-4}  describes how we can test the query containment of two CQACs using containment mappings.

\begin{theorem}
	\label{thm:cont-CQAC}
	
	Let $Q_1,Q_2$ be CQACs, and $Q'_1=Q'_{10}+\beta'_1 ,Q'_2=Q'_{20}+\beta'_2$
	be the respective queries after normalization.
	Suppose there is at least one containment  mapping from $Q'_{10}$ to $Q'_{20}$.
	Let $\mu_1, \ldots , \mu_k $ be all
	the containment mappings  from $Q'_{10}$ to $Q'_{20}$. Then $Q_2
	\sqsubseteq  Q_1$ if and only if the following logical implication $\phi$
	is true:
	$$\phi: \beta'_2  \Rightarrow \mu_1(\beta'_1) \vee \cdots \vee
	\mu_k(\beta'_1).$$
	(We refer to $\phi$  as the {\em containment entailment} in the rest of this paper.)
\end{theorem}

The following theorem  says that, if the CQACs have only closed ACs, then normalization is not necessary. For the proof see, e.g.,  \cite{2017Afrati}.

\begin{theorem}
	\label{thm:denorm}
	Consider two CQAC queries, $Q_1=Q_{10}+\beta_1$ and
	$Q_2=Q_{20}+\beta_2$ over densely totally ordered domains.  Suppose $\beta_1$
	contains only $\leq$ and $\geq $, and each of  $\beta_1$ and
	$\beta_2$ does not imply any ``='' restrictions.  Then $Q_2  \sqsubseteq
	Q_1$ if and only if
	$$\phi:\beta_2  \Rightarrow \mu_1(\beta_1) \vee \cdots \vee \mu_l(\beta_1),$$
	where $\mu_1,\ldots,\mu_l$ are all the containment mappings from
	$Q_{10}$ to $Q_{20} $.
\end{theorem}

As mentioned in the beginning of this section, there is another containment test for CQACs, which uses \textit{canonical databases} (see, e.g., in \cite{2017Afrati}). Considering a CQ $Q$, a canonical database is a database instance constructed as follows. We consider an assignment of the variables in $Q$ such that a distinct constant which is not included in any query subgoal is assigned to each variable. Then, the ground subgoals produced through this assignment define a canonical database of $Q$. Note that although there is an infinite number of assignments and canonical databases, depending on the constants selection, all the canonical databases are isomorphic; hence, we refer to such a database instance as the canonical database of $Q$.  To test, now, the containment  $Q_2\sqsubseteq  Q_1$ of the CQs $Q_1$, $Q_2$, we compute the canonical database $D$ of $Q_2$ and check if $Q_2(D)\subseteq Q_1(D)$.

Extending the test using canonical databases to CQACs, a single canonical database does not suffice. We construct a canonical database of a CQAC $Q_2$ with respect to a CQAC $Q_1$ as follows. Consider the set $S=S_V\cup S_C$ including the variables $S_V$ of $Q_2$, and the constants  $S_C$ of both $Q_1$ and $Q_2$. Then, we partition the elements of $S$ into blocks such that no two distinct constants are in the same block. Let $\CP$ be such a partition; for each block in the partition $\CP$, we equate all the variables in the block to the same variable and, if there is a constant in the block, we equate all the variable to the constant. For each partition $\CP$, we create a number of {\em canonical databases}, one for each total ordering on the variables and constants that are present (after we have equated appropriately, as explained).

Although there is an infinite number of canonical databases, depending of the constants selected, there is a bounded set of canonical databases such that every other canonical database is isomorphic to one in this set.
Such a set is referred as \textit{the set of canonical databases} of $Q_2$ w.r.t. $Q_1$. To test now the containment $Q_2\sqsubseteq  Q_1$ of the CQACs $Q_1$, $Q_2$, we construct all the canonical databases of $Q_2$ w.r.t. $Q_1$ and, for each canonical database $D$, we check if $Q_2(D)\subseteq Q_1(D)$.

\begin{theorem}
	A CQAC query $Q_2$ is contained into a CQAC query $Q_1$ if and only if, for each database belonging to the set of canonical databases of $Q_2$ with respect to $Q_1$, the query $Q_1$ computes all the tuples that $Q_2$ computes if applied on it.
\end{theorem}

\subsection{Answering queries using views}

A \textit{view} is a named query which can be treated as a regular relation. The query defining the view is called \textit{definition} of the view (see, e.g., in \cite{2017Afrati}). 

%
%
%
 Considering a set of views $\V$ and a query $Q$ over a database schema $\CS$, we want to answer $Q$ by accessing only the instances of views~\cite{Levy95-1, Halevy01-1, 2017Afrati}.  To answer the query $Q$ using $\V$ we could rewrite $Q$ into a new query $R$ such that $R$ is defined in terms of views in $\V$ (i.e., the predicates of the subgoals of $R$ are view names in $\V$). 
We denote by $\V(D)$ the output of applying all the view definitions on a database instance $D$. Thus, $\V(D)$ and any subset of it defines a view instance $\I$ for which there is a database $D$ such that $\I \subseteq \V(D)$.

If,  for every database instance $D$, we have $R(\V(D))=Q(D)$ then $R$ is an \textit{equivalent rewriting} of $Q$ using $\V$. If $R(\V(D))\subseteq Q(D)$, then $R$ is a \textit{contained rewriting} of $Q$ using $\V$. To find and check query rewitings we use the concept of  expansion   which is defined as follows.

\begin{definition}
	The \emph{view-expansion},\footnote{In Section \ref{subsec-mcr-datalogAC}, we will need to differentiate between view-expansion and Datalog-expansion which we will define shortly, therefore, when confusion arises we use these prefixes.} $R^{exp}$, 
 of a rewriting $R$ defined in terms of views in $\V$,  is obtained from $R$ as follows. For each subgoal $v_i$ of $R$ and the corresponding view definition $V_i$ in $\V$, if $\mu_i$ is the mapping from the head of $V_i$ to $v_i$ we replace $v_i$ in $R$ with the body of $\mu_i(V_i)$. The non-distinguished variables in each view are replaced with
	fresh variables in $R^{exp}$.
\end{definition}

To test  whether a query $R$ defined in terms of views set $V$  is a contained (resp. equivalent) rewriting of a query $Q$ defined in terms of the base relations, we check whether  $P^{exp} \sqsubseteq Q$ (resp. $P^{exp} \equiv Q$). 
%
%
%
%
%

There are settings where there is no equivalent rewriting of the query using the views. In such cases, finding a containing rewriting returning as many answers of the query as possible matters. In this context, we define a contained rewriting, called \emph{maximally contained rewriting}  (\emph{MCR}, for short), that returns most of the answers of the query.

\begin{definition}
	A rewriting $R$ is called a {\em maximally contained rewriting} ({\em MCR}) of query $Q$ using views $\V$ with respect to query language $\CL$ if 
	\begin{enumerate}
		\item $R$ is a contained rewriting of $Q$ using $\V$ in $\CL$, and 
		\item every contained rewriting of $Q$ using $\V$ in language $\CL$ is contained in $R$.
	\end{enumerate}
\end{definition}

A view instance $\I$  is a database with facts of the view relations. It is expected that $\I$ is computed by applying the views on a database over the base relations in terms of which the views are defined. The notion of certain answers is another way to get information from a view instance about the query.

\begin{definition}
	We define the certain answers of ($Q,\I$) with respect to $\V$ as follows:
	\begin{itemize}
		\item Under the Closed World Assumption (CWA):
		\[\text{certain}(Q,\I)=\bigcap\{Q(D): D \text{ such that } \I=\V(D)\}.\] 
		\item Under the Open World Assumption (OWA):
		\[\text{certain}(Q,\I)=\bigcap\{Q(D): D \text{ such that } \I\subseteq\V(D)\}.\] 
	\end{itemize}
	
	\end{definition}

The relation between what an MCR computes and the set of certain answers on a view instance is not easy to find. 
In sections  \ref{sec-6} and \ref{subsec-mcr-datalogAC}, we present the way MCRs and certain answers are connected for CQACs, under the OWA.

\subsection{Datalog queries}
A \textit{Datalog query} (a.k.a.  Datalog program) is a finite set of Datalog rules, where a \textit{rule} is a CQ whose predicates in the body could either refer to a base relation or to a head of a rule in the query (either the same rule or other rule). Furthermore, there is a designated predicate, which is called \textit{query predicate}, and returns the result of the query.

The predicates in the body of each rule in a Datalog query are of two types; the ones referring to base relations and the ones referring to a head of a rule. The predicates of the former type are called \textit{extensional} (\textit{EDB}, for short) while the predicates of the latter are called intensional  (\textit{IDB}, for short). The atom whose predicate is an EDB (resp. IDB) is called \textit{base atom} (resp. \textit{derived atom}). A Datalog query is called \textit{monadic} if all the IDBs are unary.

The evaluation of a Datalog query on a database instance is performed by applying the rules on the database until no more facts (i.e., ground head atoms) are added to the set of the derived atoms. The answer of a Datalog query on a database is the set of facts derived during the computation for the query predicate. Namely, the evaluation follows the fixpoint semantics. A $Datalog^{AC}$  query allows in each rule also arithmetic comparisons (ACs) as subgoals, i.e., each rule  is a CQAC. The evaluation process remains the same, only now, the AC subgoals should be satisfied too. 
We say that we \emph{unfold a rule} if we replace an  IDB subgoal with the body of another rule that has this IDB predicate in its head, and we do that iteratively. 
A \textit{partial expansion} of a Datalog query is a conjunctive query that results from unfolding the rules one or more times; the partial expansion may contain IDB predicates. A \textit{datalog-expansion} of a Datalog query is a partial expansion that contains only EDB predicates. Considering all the (infinitely many) expansions of a Datalog query we can prove that a Datalog query is equivalent to an infinite union of conjunctive queries.
An expansion of a $Datalog^{AC}$  query  is defined the same way as an expansion of a Datalog query, only now we carry the ACs in the body of each expansion we produce. Thus, in an analog way, a $Datalog^{AC}$  query  is equivalent to an infinite union of CQACs.

A \textit{derivation tree} depicts a computation of a Datalog query. Considering a fact $e$ in the answer of the Datalog query, we construct a derivation tree for this fact as follows. Each node in this tree, which is rooted at $e$, is a ground fact. For each non-leaf node $n$ in this tree, there is a rule in the query which has been applied to compute the atom node $n$ using its children facts. The leaves are facts of the base relations. Such a tree is called \textit{derivation tree} of the fact $e$.

During the computation, we use an {\em instantiated rule}, which  is a rule where all the variables have been replaced by constants.  We say that a rule is {\em fired }
if there is an instantiation of this rule where all the atoms in the body of the rule are in the currently computed  database.



\section{The algorithm to check satisfaction  of a collection of ACs}
\label{subsec-algo}

We will present{\bf\texttt ~algorithm AC-sat} which, on input a collection of ACs, checks whether there is a
satisfying  assignment, i.e., an assignment of real numbers to the variables that makes all the ACs in the collection true. If there is not  then  we say that the conjunction of ACs is false or that the collection of ACs is {\em contradictory} or is not {\em  consistent}.

We define the {\em induced directed graph } of a collection $C$ of ACs of the form $X\theta Y$ where $\theta$ is one of the $<, >, \leq, \geq, =, \neq$. We  consider that this collection is divided into two sub-collections, the collection $C_A$ including all the ACs where $\theta$ is one of the $<, >, \leq, \geq, =$ and the collection $C_B$ including all the ACs where $\theta$ is $\neq$. 
The induced directed graph is built using the ACs in $C_A$ and has nodes that are variables or constants.  There is an edge labeled $\leq$ between two nodes $n_1, n_2$ if there is an AC in the collection $C_A$
which is $n_1 \leq n_2$. There is an edge labeled $<$ between two nodes $n_1, n_2$ if there is an AC in the collection $C_A$
which is $n_1 < n_2$. (We only label edges $<$ or $\leq$ since the other direction, $>$ or $\geq$ is indicated by the direction of the edge.)
We treat each equation $X=Y$ in $C_A$ as two ACs of the form $X\leq  Y$ and $X\geq Y$ and we add edges accordingly.
Finally we add edges labeled $<$ between all the pairs of constants depending on their order.\\

{\bf\texttt   ~Algorithm AC-sat: }
We consider the induced directed graph $G$ of the collection $C$ of ACs.
 We then find all the strongly connected components of $G$. We say that an edge belongs to a strongly connected component if it joins two
nodes  in this strongly connected component.

The collection $C$ of ACs is contradictory if either of the following  is true.

\begin{description}
	\item[Case 1.] There is a strongly connected component with two distinct constants belonging to it.
	\item[Case 2.] There is a strongly connected component with an edge labeled $<$.
	\item[Case 3.] There is a $A_1\neq A_2$ AC in $C_B$ such that $A_1$ and $A_2$ belong to the same strongly connected component.
\end{description}

\begin{lemma}
	\label{lemma-prelim}
	The{\bf\texttt ~algorithm AC-sat} is a complete and sound procedure to check that a conjunction of ACs is contradictory.
	
\end{lemma}

\begin{proof}
	First we prove that this procedure is complete; i.e., we prove that if the procedure shows that the conjunction is not false then we can assign
	constants to variables to make all ACs true.
	
	Since neither Case 1 nor Case 2 happens, all strongly connected components have $\leq$ labels and at most one constant. Thus,
	we assign to each of the elements of a strongly connected component the same constant, which is either a new constant or the constant of the component, as follows: We collapse each strongly connected component to one node and the induced directed graph is reduced to an acyclic directed graph. We consider a
	topological sorting of this acyclic graph into a number of levels. We assign constants
	following this topological sorting,  so that constants in the next level are greater than the constants in the previous levels. This makes all ACs true.

	Now we prove that this procedure is sound.  Whenever the procedure stops in Cases 1 and 2 then there is no assignment that satisfies all the ACs in this strongly connected component because there is a cycle
	with either two distinct constants on it or with an edge labeled $<$. This cycle means that all the variables on it should be the same. The existence of  two distinct constants on it or of an edge labeled $<$ means that two variables on the cycle should be distinct.
	Whenever the procedure stops in Case 3, then $A_1$ and $A_2$ should be equal according to the
	strongly connected component they belong. Thus we cannot find an assignment  that satisfies also the
	AC $A_1\neq A_2$.
\end{proof}

The above algorithm is used to prove the following lemma, whose full proof can be found in \cite{Afrati19}.

\begin{lemma}
	\label{lemma-LSI-appendix-NP}
	Consider the following implication:
	$$c_1\wedge c_2 \wedge ...\Rightarrow d_1 \vee d_2 \vee ...$$
	where  the conjunction of ACs $c_1\wedge c_2 \wedge ...$ is consistent
	(i.e., it has a satisfying assignment from the set of real numbers)
	and the $d_i$'s are all closed SI (i.e., either LSI or RSI) comparisons.
	%
	Then the implication is true if and only if one of
	the following happens:
	
	(i) there is a single $d_i$ from the rhs such that
	$$c_1\wedge c_2 \wedge ...\Rightarrow d_i$$
	or
	
	(ii) there are two ACs from the rhs from which one is LSI and one is RSI, say $d_i$ and $d_j$ (we call them {\em coupling ACs} for the conjunction $c_1\wedge c_2 \wedge ...$) such that
	$$c_1\wedge c_2 \wedge ...\Rightarrow d_i \vee d_j.$$

	%
\end{lemma}

%
%


\section{Analysing the containment entailment}
\label{sec-analysis}
In this section, the first two subsections serve as an introduction to the containment entailment and its preliminary analysis. 
In the end of this section, we  define  the classes of queries we consider in later sections.

Consider the containment entailment (as in Theorem \ref{thm:cont-CQAC} or  Theorem~\ref{thm:denorm}).
$$ \beta_2  \Rightarrow \mu_1(\beta_1) \vee \cdots \vee
\mu_k(\beta_1).$$

\subsection{Containment Implications}
The right hand side (rhs, for short) of the containment entailment is a disjunction of disjuncts, where each disjunct is a
conjunction of ACs. We can turn this, equivalently, to a conjunction of conjuncts, where each conjunct
is a disjunction of ACs. We call each of these last conjuncts a {\em rhs-conjunct} (from right hand side conjunct). Now we can turn the containment entailment, equivalently,  into a number of implications. In each implication, we keep the left hand side of the containment entailment the same and have the right hand side be one of the
rhs-conjuncts. We call each such implication  a {\em containment implication}.


\begin{example}
	\label{ex-firstLSI}
	For an example, consider the following normalized CQACs.
\vspace*{-.5cm}
	%
	\begin{center}
		\begin{tabular}{lll}
			$Q_1: q()$ & $\symif$ &$a(X_1,Y_1,Z_1), X_1=Y_1, Z_1< 5$ \\
			$Q_2: q()$ & $\symif$ &$a(X,Y,Z'), a(X',Y',Z), X\leq 5, Y\leq X, Z\leq Y, $$ $$X'=Y', Z'<5$\\
		\end{tabular}
	\end{center}
	
	Testing the containment $Q_2
	\sqsubseteq  Q_1$, it is easy to see that there are the following containment mappings:
\vspace*{-.5cm}
	\begin{itemize}
		\item $\mu_1:  X_1\rightarrow X, Y_1\rightarrow Y, Z_1\rightarrow Z'$
		\item $\mu_2:  X_1\rightarrow X', Y_1\rightarrow Y', Z_1\rightarrow Z$
	\end{itemize}
	
	Hence, the containment entailment is given as follows:
	
	\begin{center}
		\begin{tabular}{l}
			$X\leq 5  \wedge Y\leq X\wedge  Z\leq Y\wedge    X'=Y'\wedge  Z'<5\Rightarrow$\\
			$\big(\; \mu_1(X_1)\!\!=\!\!\mu_1(Y_1) \;\wedge\; \mu_1(Z_1)\!\!<\!5\;\big)\;\vee$ \\
			$\big(\;\mu_2(X_1)\!\!=\!\!\mu_2(Y_1) \;\wedge\; \mu_2(Z_1)\!\!< 5\;\big)$\\
		\end{tabular}
	\end{center}
	
	which is equivalently written:
	
	\begin{center}
		\begin{tabular}{l}
			$X\leq 5  \wedge Y\leq X\wedge  Z\leq Y\wedge  X'=Y'\wedge  Z'<5  \Rightarrow$\\
			$(X=Y \wedge Z'< 5)  \vee (X'=Y' \wedge Z< 5)$\\
		\end{tabular}
	\end{center}
	

	
	It is easy to verify that the above implication is true (due to the second part of the disjunction in the right-hand side which is also included in the antecedent).
%
%
%
%
%
%
%
	%
	%
	%
	%
	%
	%
	%
	%
	%
	%
	Now we consider the containment entailment we built above. According to
	what we analyzed in this section, we can equivalently rewrite this containment entailment  by
	transforming its right hand side into a conjunction, where each conjunct is a disjunction of ACs. The transformed entailment is the following, where $\beta=X\!\leq \!5  \wedge Y\!\leq \!X\wedge  Z\!\leq \!Y\wedge  X'\!\!=\!\!Y'\wedge  Z'\!<\!5$:
\vspace*{-.5cm}
	\begin{center}
		\begin{tabular}{lll}
			$\beta \Rightarrow \; (X\!\!=\!\!Y\vee X'\!\!=\!\!Y')\wedge(X\!\!=\!\!Y\vee Z\!< \!5)\wedge(Z'\!< \! 5\vee X'\!=\! Y')\wedge(Z'\!< \! 5\vee Z\!<\! 5)$ \\
		\end{tabular}
	\end{center}
	
	%
	%
	%
	%
	%
	%
	%
	%
\end{example}

%
%

The following two theorems are proved in  \cite{Afrati19}  
and serve as an introduction to the results in the present paper (the second theorem is proven based on the first theorem):


\begin{theorem}
The following two are equivalent:

a) One disjunct in the rhs suffices to make the containmnent entailment true.

b) For each containment implication, one disjunct in the rhs suffices to make it true.

\end{theorem}

\begin{theorem}
	If the containing query contains only closed  LSIs and the contained query any closed AC then the containment problem is in NP.
\end{theorem}

\subsection{ACs over single-mapping variables}
\label{sec:single-mapping}

Suppose two CQACs $Q_1=Q_{10}+\beta_1$ and $Q_2=Q_{20}+\beta_2$. 
We consider
the containment entailment:

\begin{equation} \label{eq:cont-entail}
\beta_2\Rightarrow\mu_1(\beta_1)\vee\cdots\vee\mu_k(\beta_1)
\end{equation}
\noindent
where $\mu_1,\dots, \mu_k$ are all the containment mappings from$Q_{10}$ to $Q_{20}$.
Suppose  $\beta_1$ is such that $\beta_1 = \beta_{11}\wedge\beta_{12}$ where $\beta_{11}$ is the conjunction of ACs among and on the distinguished variables and $\beta_{12}$ the ACs on the nondistinguished variables, i.e., there is no AC between a head variable and a nondistinguished variable. In this special case, we observe that in the containment entailment, each term on the right hand side becomes:
\begin{align*}
\mu_i(\beta_1)=\mu_i(\beta_{11})\wedge\mu_i(\beta_{12}).
\end{align*}
However, $\mu_i(\beta_{11})$ is the same for every term on the right hand side of the entailment because all the containment mappings $\mu_i$ are the same as concerns the distinguished variables, by definition. Thus, applying the distributive law, we write the containment entailment:
\begin{align*}
\beta_2\Rightarrow\mu_1(\beta_{11})\wedge[\mu_1(\beta_{12})\vee\cdots\vee\mu_k(\beta_{12})].
\end{align*}
Consequently, the containment entailment is equivalent to conjunction of  the following two entailments:
\begin{align*}
\beta_{2}&\Rightarrow\mu_1(\beta_{11}). \\
\beta_{2}&\Rightarrow\mu_1(\beta_{12})\vee\cdots\vee\mu_k(\beta_{12}).
\end{align*}
Hereon, we will call the second entailment the \textit{body containment  entailment} (or simply containment entailment when confusion does not arise) and the first the \textit{head entailment}.

\subsubsection{Introducing single-mapping variables}
The above  analysis  is valid because of the fact that the variables in the head of $Q_1$ always map to the same variable in $Q_2$, independently of the containment mapping from $Q_1$ to $Q_2$. 
Such a property could be straightforwardly extended to other cases. There may exist more variables (besides the head variables) of the containing query that are
 always mapped on the same variables of the contained query, for any containment mapping.  We call them single-mapping variables and give the formal definition below.

\begin{definition}(single-mapping variables)
	\label{defn:single-map-vars}
	Let $Q_1=Q_{10}+\beta_1$, $Q_2=Q_{20}+\beta_2$ be two CQACs, such that there is at least one containment mapping from $Q_{10}$ to $Q_{20}$. Consider the set ${\cal M}$ of all the containment mappings from $Q_{10}$ to $Q_{20}$.  Each variable $X$ of $Q_1$ which is always mapped on the same variable of $Q_2$ (i.e., for each $\mu\in\CM$ the $\mu(X)$ always equals the same variable) is called a \textit{single-mapping variable with respect to $Q_2$}.
\end{definition}

Notice that the head variables of $Q_1$ are single-mapping variables with respect to any query. For another example, consider that there is a predicate $r$ such that  $Q_1$ has $g_{11}, g_{12},\dots , g_{1n}$ subgoals with predicate $r$ and $Q_2$ has a \textit{single} subgoal $g_2$ with predicate $r$. Since each of the $g_{11}, g_{12},\dots , g_{1n}$ subgoals maps on $g_2$, for every containment mapping from $Q_1$ to $Q_2$, the variables in $g_{11}, g_{12},\dots , g_{1n}$ subgoals are single-mapping variables.

Thus, we  extend the previous analysis  in the next subsection and show that the containment entailment can be decomposed into two parts in a more general case.

\subsection{The classes of queries} 
Here we define what it means for a pair of CQACs to be a  disjoint-AC pair wrto a set of single-mapping variables. 
Then we state  Proposition 	\ref{prop:single-map-vars} that says, that, for such a pair, the containment entailment can be broken in two entailments. Then we restrict our definition to containing queries that 
only allow SI on non-single-mapping variables and, in particular with only one RSI. This is the class of queries for which we prove in Section \ref{ssec:mot} that the containment test is in NP. We also define CQAC queries which we call RSI1+ queries and this is the class of queries for which the results of Section \ref{subsec-mcr-datalogAC} hold.
\begin{definition}
Let $Q_1=Q_{10}+\beta_1$, $Q_2=Q_{20}+\beta_2$ be two CQACs with closed ACs, such that there is at least one containment mapping from $Q_{10}$ to $Q_{20}$. Let ${\cal X}_1$ be the set of variables of $Q_1$. We assume  that the set , ${\cal X}_1$, of variables of $Q_1$ can be partitioned into the sets ${\cal X}_1^{sv}$, ${\cal X}_1^{nsv}$, s.t. ${\cal X}_1^{sv}\cap{\cal X}_1^{nsv}=\emptyset$, ${\cal X}_1^{sv}$ contains only single-mapping variables of $Q_1$ with respect to $Q_2$ and there are no ACs of $Q_1$ joining a variable in ${\cal X}_1^{sv}$ with a variable in ${\cal X}_1^{nsv}$. Then we say that ($Q_1$,$Q_2$) is a {\em disjoint-AC pair} with respect to ${\cal X}_1^{sv}$\footnote{When it is obvious from the context, we do not refer to  ${\cal X}_1^{sv}$.}. 
\end{definition}

\begin{proposition}
	\label{prop:single-map-vars}
	Let $Q_1=Q_{10}+\beta_1$, $Q_2=Q_{20}+\beta_2$ be two CQACs with closed ACs, such that there is at least one containment mapping from $Q_{10}$ to $Q_{20}$. Let ${\cal X}_1$ be the set of variables of $Q_1$.
Let ($Q_1$,$Q_2$) be a {\em disjoint-AC pair} with respect to ${\cal X}_1^{sv}$, where ${\cal X}_1^{sv}$ contains only single-mapping variables of $Q_1$ with respect to $Q_2$.

Then, the containment entailment $\beta_2\Rightarrow\mu_1(\beta_1)\vee\cdots\vee\mu_k(\beta_1)$  is true if and only if both the following two are true:
	
	\begin{itemize}
		\item 	$\beta_{2}\Rightarrow\mu_1(\beta_{11})$,  {\em head containment entailment} and
		\item $\beta_{2}\Rightarrow\mu_1(\beta_{12})\vee\cdots\vee\mu_k(\beta_{12})$, {\em body containment entailment}\footnote{We retain the same names as in the simple case above, for simplicity of reference; they are actually  {\em single-mapping entailment } and {\em non-single-mapping entailment}.}
	\end{itemize}
	where  $\mu_1,\dots \mu_k$ are all the containment mappings from $Q_{10}$ to $Q_{20}$ and $\beta_1 = \beta_{11}\wedge\beta_{12}$,
	where $\beta_{11}$
	includes all the ACs of $\beta_1$ over the variables in ${\cal X}_1^{sv}$, and $\beta_{12}$ includes all the ACs of $\beta_1$  over the variables in ${\cal X}_1^{nsv}={\cal X}_1 - {\cal X}_1^{sv}$.
\end{proposition}

The proof of the Proposition~\ref{prop:single-map-vars} is an immediate consequence of the Definition of single-mapping variables and was analysed in details in the previous section. 

%
%



%
%

We define a CQAC CRSI1+ query, or simply RSI1+ query hereon, to be a query that:
\begin{enumerate}
	\item  It has only closed ACs.
	\item There are no ACs between a head variable and a nondistinguished variable.
	\item The ACs on  nondistinguished variables are semi-interval ACs  and there is a single right semi-interval AC.
\end{enumerate}
When there are no ACs on the head variables, then we say that this is a RSI1 query.

Notice that, given a query $Q_1$ which is a RSI1+ and any CQAC query $Q_2$ then the pair $(Q_1,Q_2)$ is  a disjoint-AC pair.
%
%

The following definitions formally describes the {\em RSI1 disjoint-AC pair}.
\begin{definition} Let $V_{sm}$ be a set of single-mapping variables in $Q_1$ and $V_{Q_1}$ be the set of variables of $Q_1$.
	A pair of CQACs  ($Q_1$, $Q_2$) is called {\em RSI1 disjoint-AC pair with respect to $V_{sm}$} if the following is true:
	\begin{enumerate}
		\item  Both $Q_1$ and $Q_2$ have only closed ACs.
		\item There are no ACs between the variables in $V_{sm}$ and the variables in $V_{Q_1}-V_{sm}$.
		\item The ACs in $Q_1$ are such that the following are true: 
		\begin{enumerate}
			\item The ACs on variables in $V_{Q_1}-V_{sm}$ are semi-interval (SI, for short),  and 
			\item there is a single right semi-interval (RSI) AC, among the ACs on variables in $V_{Q_1}-V_{sm}$.
		\end{enumerate}
	\end{enumerate}
\end{definition}

Notice that, given a query $Q_1$ which is a RSI1+ and any CQAC query $Q_2$ then the pair $(Q_1,Q_2)$ is  an RSI1 disjoint-AC pair.

We say that a body containment entailment is an {\em RSI1 entailment} if the ACs in each disjunct on the right hand side include only one RSI AC and the others are LSI ACs.

\begin{itemize}
	\item For every RSI1 disjoint-AC pair, the body containment entailment is an RSI1 entailment.
	\item In the next section, we consider  RSI1 disjoint-AC pairs of queries.
\end{itemize}

Naturally, because of symmetry, we can define LSI1 disjoint-AC pairs of quries where now only one LSI is allowed and all the results are also valid for this class.


%
%


\section{CQAC Query Containment Using Datalog}
\label{ssec:mot}

The main result of this section  is the following theorem:

\begin{theorem}
	\label{thm-np-complete}
	Consider a pair ($Q_1$, $Q_2$) which is a  RSI1 disjoint-AC pair of queries.
	Then testing containment of $Q_2$ to $Q_1$ is NP-complete.
\end{theorem}

A byproduct of the proof of this theorem is a reduction of the CQAC containment problem, in this special case, to a containment problem where we check containment of a CQ to a Datalog query (i.e.,  both these queries have no ACs, their definitions use only relational atoms). This reduction is also important in other sections of this paper where we use it to construct MCRs for CQAC queries and views and prove that certain answers can be computed in polynomial time for certain cases of queries and views.

 Proposition~\ref{prop:single-map-vars} leads us to focus on the body containment entailment of the two CQAC queries. Thus, we ignore the ACs of the containing query that are on the single-mapping variables and call the resulting query the {\em reduced containing query}. For the first three subsections of this section, we will only refer to the reduced containing query, so,  we will say simply containing query. Note, here, that we do not ignore any AC from the contained query, since all the ACs of the contained query are required in order to check body containment entailment.

Thus this section has two large parts:
\begin{itemize}
	\item Transformation of the reduced containing query $Q_1$  to a Datalog query and transformation of the contained query $Q_2$  into a CQ query. This is presented in the three first subsections of this section.
	\item Proving that $Q_2$ is contained in the reduced containing query $Q_1$ if and only if their transformed CQ and Datalog queries, respectively, are contained in each other. The main results of this section are stated formally  in 
Subsection \ref{subsec-main-proof-sec5}.
\end{itemize}

Theorem \ref{thm-np-complete} extends significantly the corresponding result in  \cite{AfratiLM06}.
The transformations and the proof are along similar lines as the transformations  and the proof explained in \cite{AfratiLM06} with many modifications to capture the new features.
Algorithm  AC-sat presented in Section \ref{subsec-algo} is missing from \cite{AfratiLM06}.
This algorithm  offers an elegant way to prove  technical  preliminary results about implications involving arithmetic comparisons. 

The structure of this section is as follows: In Subsection
\ref{subsec-tree-like}, two implications are analyzed that will be met when we prove the main result in this section later on. These implications are simply implications that involve ACs, their relation to the containment problem is that they have the same structure as the containment entailment. Thus, we provide some explanations in Subsection
\ref{subsec-tree-like} as to the reason these results lead towards the idea of using the Datalog transformation of the containing query.
Then in Subsections \ref{subsec-contained-trans} and
\ref{subsec-construct-Datalog} we present the transformations of the reduced containing query and the contained query respectively.

Subsection \ref{subsec-main-proof-sec5} contains the statement of the main results in this section. 
Subsection \ref{subsec-examples} contains examples of the transformations presented in Subsections \ref{subsec-contained-trans} and
\ref{subsec-construct-Datalog}.

Finaly, in Subsection
\ref{subsec-simplefacts}, we present preliminary partial results and intuition for the proof of the main technical result. The proof itself is presented in \ref{prf:thm-main123p}.
%

\subsection{The tree-like structure of the containment entailment}
\label{subsec-tree-like}

First, as Theorem 2.4 shows, the query normalization is not needed for testing containment into this setting.
The following proposition is where the class of RSI1s comes useful.


\begin{proposition}
	\label{trick-pro}
	Let $\beta$ be a conjunction of closed ACs which is consistent, and each $\beta_1, \beta_2, \ldots ,\beta_k$  be a conjunction of closed
	RSI1s. Suppose the following is true:
	$$\beta \Rightarrow \beta_1 \vee \beta_2 \vee \cdots \vee \beta_k.$$
	Then there is a $\beta_i$ (w.l.o.g. suppose it is $\beta_1$) such that either  of the following two happens:
	
	%
	%
	%
	%
	%
	%
	
	\begin{enumerate}[label=(\roman*)]
		\item $\beta \Rightarrow \beta_1,$ \textbf{or}
		\item there is an AC $e$ in $\beta_1$ such that the following are true:
		\begin{enumerate}[label=(\alph*)]
			\item $\beta \wedge \neg e \Rightarrow   \beta_2 \vee \cdots \vee \beta_k$
			(or equivalently,
			$\beta  \Rightarrow   \beta_2 \vee \cdots \vee \beta_k \vee e$),
			\item $\beta \Rightarrow \beta_1 \vee \neg e$, and
			\item all the other ACs, besides $e$,  in $\beta_1$ are directly implied by $\beta$.
		\end{enumerate}
	\end{enumerate}
	
\end{proposition}
\begin{proof}
	Suppose there is no $\beta_i$ such that $$\beta \Rightarrow \beta_i$$
	Then we claim that there is
	a $\beta_i$ (w.l.o.g.  suppose it is $\beta_1$) such that all the ACs in $\beta_1$ are directly implied by $\beta$ (i.e., $\beta \Rightarrow e_i$ if $e_i$ is an AC in $\beta_1$), except for one AC $e$.,  i.e., we claim that also the following is true:
	$$\beta \Rightarrow \beta_1 \vee \neg e$$
	Towards  contradiction,
	suppose that for all the $\beta_i$s there are at least two ACs that are not directly implied by $\beta$. Since all
	the $\beta_i$'s are RSI1s, each $\beta_i$ has at least one LSI that is not directly implied. If we take all these LSI's after applying the distributive law and converting the right-hand side from a disjunction of conjunctions to a conjunction of disjunctions, then we will have a conjunct that contains only LSIs, none of which is directly implied by $\beta$.
	We  need to show that this is impossible --- i.e., it is not true that $\beta \Rightarrow ac_1 \vee ac_2 \cdots $ if
	none of the LSI $ac_i$ is directly implied by $\beta$. This is proved in Lemma \ref{lemma-LSI-appendix-NP}.
	
	
	Now we write equivalently the implication in the statement of the proposition as:
	$$\beta \wedge \neg \beta_1\imply
	\beta_2 \vee \beta_3 \cdots \vee \beta_k,$$ or
	equivalently (assuming $\beta_1=e_1 \wedge \cdots \wedge e_t$,
	where the $e_i$s are ACs)
	$$(\beta \wedge \lnot e_1)
	\vee (\beta \wedge \lnot e_2) \vee \cdots \vee
	(\beta \wedge \lnot e_t) \imply \beta_2 \vee \beta_3 \cdots \vee \beta_k.$$
	Assume w.l.o.g. that $e=e_1$. Since each $e_i$, with the exception of $e_1$, is
	entailed by $\beta$, each disjunct with the exception of the first one in the left-hand side
	is always false. Hence, the latter entailment yields:\\
	$~~~~~~~~~~~~~~~~~~~~~~~~~~~~~~~~~~~~~~\beta \wedge \neg e\imply \beta_2 \vee \beta_3 \cdots
	\vee \beta_k.$
\end{proof}

%
%

Proposition
\ref{trick-pro}  begins to show a tree-like structure of the containment entailment and it gives the first intuition for constructing a Datalog query from the containing query that will help in deciding query containment. The following example gives an illustration of this intuition.

\begin{example}
	\label{ex:description-of-prop52}
	Let us consider the following two Boolean queries.
	\begin{center}
		\begin{tabular}{lll}
			$Q_1: q()$ & $\symif$ &$a(X,Y,Z),X\leq 8,Y\leq 7,Z\geq 6.$ \\
			$Q_2: q()$ & $\symif$ &$a(X,Y,Z),a(U_1,U_2,X),a(V_1,V_2,Y),$\\
			& &$ a(Z,Z_1,Z_2),a(U_1',U_2',U_1),a(V_1',V_2',V_1),$\\
			& &$ U_1'\leq 8,U_2'\leq 7,U_2\leq 7,V_1'\leq 8,$\\
			& &$ V_2'\leq 7,V_2\leq 7,Z_1\leq 7,Z_2\geq 6.$\\
		\end{tabular}
	\end{center}
	
	\begin{figure}
		\centering
		\includegraphics[width=0.5\linewidth]{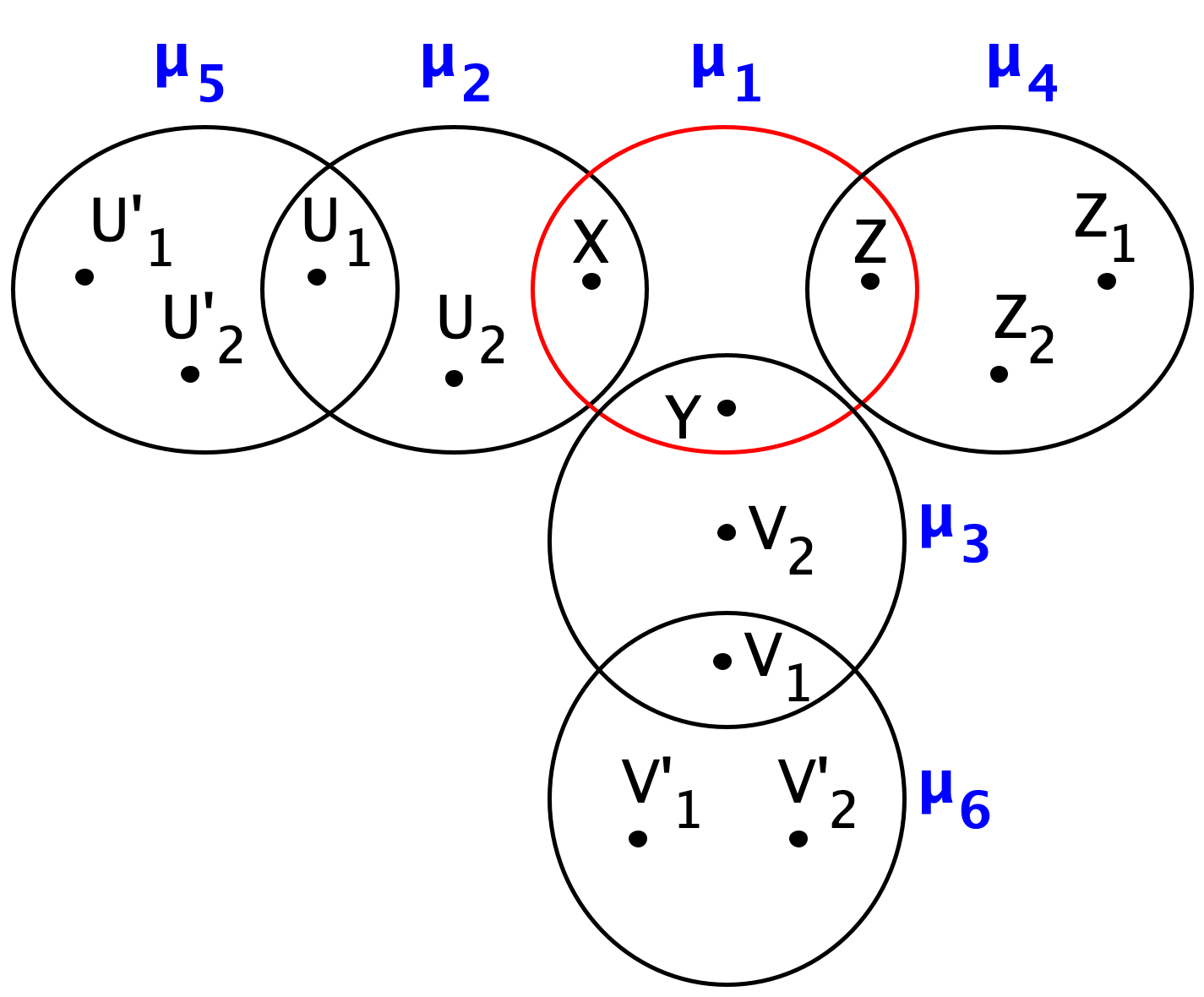}
		\caption{Illustration of containment entailment of Example~\ref{ex:description-of-prop52}}
		\label{fig:intuition-bodyq2}
	\end{figure}
	
	The query $Q_2$ is contained in the query $Q_1$. To verify this, notice that there are $6$ containment mappings from\footnote{we always mean containment mappings from the relational subgoals of $Q_1$ to the relational subgoals of $Q_2$} $Q_1$ to $Q_2$. These mappings are given as follows: $\mu_1:$ $(X,Y,Z)\rightarrow(X,Y,Z)$, $\mu_2:$ $(X,Y,Z)\rightarrow(U_1,U_2,X)$, $\mu_3:$ $(X,Y,Z)\rightarrow(V_1,V_2,Y)$, $\mu_4:$ $(X,Y,Z)\rightarrow(Z,Z_1,Z_2)$, $\mu_5:$ $(X,Y,Z)\rightarrow(U_1',U_2',U_1)$, and $\mu_6:$ $(X,Y,Z)\rightarrow(V_1',V_2',V_1)$. After replacing the variables as specified by the containment mappings, the query entailment is $\beta\Rightarrow\beta_1\vee\beta_2\vee\beta_3\vee\beta_4\vee\beta_5\vee\beta_6$, where:
	
	\begin{center}
		\begin{tabular}{ll}
			\multicolumn{2}{l}{$\beta:$ $U_1'\leq 8\;\wedge U_2'\leq 7\;\wedge U_2\leq 7\;\wedge V_1'\leq 8\;\wedge V_2'\leq 7\;\wedge V_2\leq 7\;\wedge Z_1\leq 7\;\wedge Z_2\geq 6$.}\\
			$\beta_1:$ $X\leq 8\;\wedge Y\leq 7\;\wedge Z\geq 6$. & $\beta_4:$ $Z\leq 8\;\wedge Z_1\leq 7\;\wedge Z_2\geq 6$.\\
			$\beta_2:$ $U_1\leq 8\;\wedge U_2\leq 7\;\wedge X\geq 6$. & $\beta_5:$ $U_1'\leq 8\;\wedge U_2'\leq 7\;\wedge U_1\geq 6$.\\
			$\beta_3:$ $V_1\leq 8\;\wedge V_2\leq 7\;\wedge Y\geq 6$. & $\beta_6:$ $V_1'\leq 8\;\wedge V_2'\leq 7\;\wedge V_1\geq 6$.\\
		\end{tabular}
	\end{center}
	
%
	We now refer to Figure~\ref{fig:intuition-bodyq2} to offer some intuition about  and visualization on Proposition \ref{trick-pro} using the above queries. The circles in the figure represent the mappings $\mu_1, \dots, \mu_6$, and the dots are the variables of $Q_2$. Notice now the intersections between the circles.
	Proposition \ref{trick-pro} refers to these intersections, such as the one between $\mu_3$ and $\mu_6$ (or,  the one between $\mu_2$ and $\mu_5$).
	
	The AC $V_1\geq 6$ ($V_1$ is included in the intersection between $\mu_3$ and $\mu_6$) is the one that is not directly implied  by $\beta$, as stated in the case (ii) of the Proposition \ref{trick-pro}. In particular, it is easy to verify that the following are true:
	
	\begin{itemize}
		\item $\beta\wedge \neg (V_1\geq 6)\Rightarrow\beta_1\vee\beta_2\vee\beta_3\vee\beta_4\vee\beta_5$.
		\item $\beta\Rightarrow\beta_6\vee \neg (V_1\geq 6)$ (i.e., $\beta\Rightarrow(V_1'\leq 8\;\wedge V_2'\leq 7\;\wedge V_1\geq 6)\vee \neg (V_1\geq 6)$).
		\item $\beta\Rightarrow(V_1'\leq 8)$ and $\beta\Rightarrow(V_2'\leq 7)$.
	\end{itemize}

	
\end{example}

Proposition \ref{trick-pro1} is a generalization of
Proposition
\ref{trick-pro}.

\begin{proposition}
	\label{trick-pro1}
	Let $\beta$ be a conjunction of closed SI ACs which is consistent, and $\beta_1, \beta_2, \ldots ,\beta_k$ each be a conjunction of closed
	RSI1s (i.e., in each conjunct there is only one RSI and the rest are LSI ACs). Suppose the following is true:
	$$\beta \Rightarrow \beta_1 \vee \beta_2 \vee \cdots \vee \beta_k \vee e_1 \vee e_2 \vee \cdots$$
	where $e_i$s are closed SIs such that the following implication is not true: $\beta \Rightarrow e_1\vee e_2 \vee \cdots$.
	Then there is a $\beta_i$ (w.l.o.g. suppose it is $\beta_1$) such that either of the following two happen:
	
	\begin{enumerate}[label=(\roman*)]
		\item  $\beta \Rightarrow \beta_1 \vee e_1 \vee e_2 \vee \cdots,$ \textbf{or}
		\item there is an AC $e$, called {\em special for this mapping}, in $\beta_1$ such that the following  are true:
		\begin{enumerate}[label=(\alph*)]
			\item $\beta \wedge \neg e \Rightarrow   \beta_2 \vee \cdots \vee \beta_k \vee e_1 \vee e_2 \vee \cdots$
			, or equivalently,
			$$\beta  \Rightarrow   \beta_2 \vee \cdots \vee \beta_k \vee e \vee e_1 \vee e_2 \vee \cdots.$$
			\item $\beta \Rightarrow \beta_1 \vee \neg e \vee e_1 \vee e_2 \vee \cdots $  and
			\item all the other ACs $ac_j$ in $\beta_1$,  with $j=1,2,\ldots $,  besides $e$, are either  directly implied by $\beta$ or coupled with one of the $e_i$s for $\beta$ i.e., either $\beta \Rightarrow ac_j$ or
			$\beta\Rightarrow e_i \vee ac_j$.
		\end{enumerate}
	\end{enumerate}

\end{proposition}

\begin{proof}
	
	Suppose there is no $\beta_i$ such that $$\beta \Rightarrow \beta_i\vee e_1\vee e_2 \vee \cdots $$
	Then we claim that there is
	a $\beta_i$ (w.l.o.g.  suppose it is $\beta_1$) such that  all the ACs $a_i$  in $\beta_1$ are such that  $a_i\vee e_1\vee \cdots$ is directly implied by $\beta$ (i.e., $\beta \Rightarrow a_i\vee e_1\vee \cdots$ if $a_i$ is an AC in $\beta_1$), except for one AC $a_1=e$ (wlog suppose this is $a_1$),  i.e., we claim that  the following is true for $e$:
	$$\beta \Rightarrow \beta_1 \vee \neg e\vee e_1\vee e_2 \vee \cdots$$
	Towards  contradiction,
	suppose that for all the $\beta_i$s there are at least two ACs  (say  AC $a^{i12}$ is such an AC) such that the following does not happen:
	\begin{equation} \label{eq-1a}
	\beta \Rightarrow    a^{i12}\vee e_1\vee e_2 \vee \cdots
	\end{equation}
	\noindent
	
	Since all
	the $\beta_i$'s are RSI1s, each $\beta_i$ has at least one LSI for which the implication \ref{eq-1a} is not true. If we take all these LSI's (after applying the distributive law and converting the right-hand side from a disjunction of conjunctions to a conjunction of disjunctions), then we will have a conjunct that contains only LSIs, none of which is such that
	the implication \ref{eq-1a} is true.
	Then we will have a case like in Lemma \ref{lemma-LSI-appendix-NP}.
	According to Lemma \ref{lemma-LSI-appendix-NP}, there are two cases: a) There is a single SI on the rhs  which is implied by $\beta$ or b) there are two SI in the rhs whose disjunction is implied, of which one is LSI and one is RSI. Thus, in both cases, we have only one LSI, say it is $a_{LSI}$ such that
	$$\beta\Rightarrow a_{LSI} \vee e_1\vee e_2 \vee \cdots.$$
	This is a contradiction to our assumption.

	We write equivalently the implication in the statement of the proposition as:
	$$\beta \wedge \neg [\beta_1\vee e_1\vee e_2 \vee \cdots ]\imply
	\beta_2 \vee \beta_3 \cdots \vee \beta_k$$ or
	equivalently (assuming $\beta_1=a_1 \wedge \cdots \wedge a_t$,
	where the $e_i$s are ACs)
	$$(\beta \wedge \neg a_1\wedge \neg e_1 \wedge \neg e_2 \wedge \cdots)
	\vee (\beta \wedge \neg a_2\wedge \neg e_1 \wedge \neg e_2 \wedge \cdots) \vee \cdots \vee
	(\beta \wedge \neg a_t\wedge \neg  e_1 \wedge \neg e_2 \wedge \cdots)$$ $$ \imply \beta_2 \vee \beta_3 \cdots \vee \beta_k$$
	Assume w.l.o.g. that $e=a_1$. Since each $a_i\vee e_1\vee \cdots$, with the exception of $a_1$, is
	entailed by $\beta$, each disjunct with the exception of the first one in the left-hand side
	is always false. Hence, the latter entailment yields:\\
	~~$~~~~~~~~~~~~~~~~~~~~~~~~~~~~~~~~~~~\beta \wedge \neg e\imply \beta_2 \vee \beta_3 \cdots
	\vee \beta_k\vee e_1\vee e_2 \vee \cdots.$
\end{proof}

\begin{example}
	\label{ex:trick-pro1}
	Continuing Example~\ref{ex:description-of-prop52}, we will use the Proposition \ref{trick-pro1} to see how the AC on the variable $V_1$ is related to other ACs on the same variable in another mapping (here it is the mapping $\mu_4$). To see that, notice that the AC $V_1\geq 6$ is the special AC for $\mu_6$ and it is coupled with the AC $V_1\leq 8$ in $\beta_3$ (i.e., $\beta\Rightarrow (V_1\geq 6) \vee (V_1\leq 8)$). In particular, as we saw in  Example~\ref{ex:description-of-prop52}, the following is true.
	$$\beta\Rightarrow\beta_1\vee\beta_2\vee\beta_3\vee\beta_4\vee\beta_5\vee (V_1\geq 6).$$
	
	Then, according to the Proposition~\ref{trick-pro1} (where $e_1=V_1\geq 6$), there is $\beta_i$ (in this case, $\beta_3$ is such a $\beta_i$) such that the following are true (case (ii) in the proposition):
	\begin{itemize}
		\item $\beta\wedge \neg (Y\geq 6)\Rightarrow\beta_1\vee\beta_2\vee\beta_4\vee\beta_5\vee (V_1\geq 6)$.
		\item $\beta\Rightarrow\beta_3\vee \neg (Y\geq 6)\vee (V_1\geq 6)$; i.e., 
		$$\beta\Rightarrow(V_1\leq 8\;\wedge V_2\leq 7\;\wedge Y\geq 6)\vee \neg (Y\geq 6)\vee (V_1\geq 6).$$
		\item $\beta\Rightarrow(V_2\leq 7)$, while $V_1\leq 8$ is coupled with $V_1\geq 6$.
	\end{itemize}
	
	
\end{example}

We give a first glance of what is going to happen in the rest of this section. In particular, we do  the following:

%
%
%

\begin{enumerate}
	\item We transform the containing query $Q_1$ into a Datalog query $Q^{Datalog}_{Q_1}$.
	\item We transform the contained query into a CQ, $Q^{CQ}_{Q_2}$. 
	\item The above two transformations are done by keeping the relational subgoals of $Q_1$ ($Q_2$, respectively) and encoding  the
	arithmetic comparisons into  relational predicates.
	\item
	We prove (Theorem \ref{thm:main}) that $Q_2$ is contained in $Q_1$ if and only if  $Q^{CQ}_{Q_2}$ is contained in $Q^{Datalog}_{Q_1}$.
\end{enumerate}

Intuitively, using those transformations we aim to replace the ACs with relations; hence, transform the problem of CQAC containment to a containment problem of a Datalog query in a CQ. One might wonder why the transformation of the containing query to a Datalog query is required. The answer to this question is based on the containment entailment. The disjunction in the right-hand-side implies arbitrary combinations of the ACs, since the contained query can be arbitrarily long independently of the size of the containing query. Hence, the program-expansion of $Q^{Datalog}_{Q_1}$ that verifies the containment can be 
 arbitrarily long, depending on the size of the contained query.

\subsection{Construction of Datalog Query for Containing Query}
\label{subsec-construct-Datalog}
In this subsection, we describe the construction of a Datalog query for a given RSI1 query $Q$.  

The Datalog query has two kinds of rules: The rules that depend only on the containing query, and we call them {\em basic rules},  and the rules that also take into account the contained query, and we call them {\em dependant rules}.

In various places, in order to illustrate the construction, we will use the query in the following running example.

\begin{example}
	\label{ex:running41}
	The following query $Q_1$ is an  RSI1 query:
	\begin{center}
		\begin{tabular}{lll}
			$Q_1({W_1,W_2})$ & $\symif$ &$a(W_1,W_2,Y), e(X,Y),e(Y,Z),X\geq 5,Z\leq 8.$ \\
		\end{tabular}
	\end{center}
	For simplicity in the notation we will denote by $\overline{W}$  the vector $W_1,W_2$ of head variables. Thus, we are writing the query as:
	\begin{center}
		\begin{tabular}{lll}
			$Q_1(\overline{W})$ & $\symif$ &$a(\overline{W},Y),e(X,Y),e(Y,Z),X\geq 5,Z\leq 8.$ \\
		\end{tabular}
	\end{center}
\end{example}

{\bf Construction of the basic rules} $Q_1^{Datalog}:$ We construct three kinds
of rules, {\em mapping rules, coupling rules,} and a single {\em query rule}.

First, we introduce the EDB predicates and the IDB predicates that we use and describe how we construct them.
The EDB predicates are all the predicates from the relational subgoals of $Q_1$ and an extra binary predicate $U$. Intuitively, $U(X,Y)$ encodes the AC $X\leq Y$.
Now, the IDB predicates are as follows:

\begin{enumerate}
	\item
	We introduce new semi-unary IDBs,\footnote{We call them semi-unary for reasons that will become apparent later during the proof.}  two pairs of IDBs for each
	constant $c$ in $Q_1$ (intuitively, that compares a non-single-mapping variable to this constant), namely $I_{\geq c}$, $I_{\leq c}$ and
	$J_{\geq c}$, $J_{\leq c}$. Intuitively, these predicates have as arguments the vector $\overline{W}$ of variables in the head of the query $Q_1$ and another variable $X$.
	\item
	For each AC $X\theta c$, we construct  the IDB predicate atoms  $I_{\theta
		c}(X,\overline{W})$  and  $J_{\theta c}(X,\overline{W})$, where $\theta$ is either $\leq$ or $\geq$.
	
	\item
	For each AC $X\theta c$, considering  the IDB predicate atom $I_{\theta
		c}(X,\overline{W})$ ($J_{\theta c}(X,\overline{W})$, respectively), we refer to
	$J_{\theta c}(X,\overline{W})$
	($I_{\theta
		c}(X,\overline{W})$, respectively),
	as the {\em associated $I$-atom } ( {\em associated $J$-atom }
	respectively) of $X\theta c$, and we refer to  $X\theta c$  as the {\em associated AC} of
	$I_{\theta
		c}(X,\overline{W})$ ($J_{\theta c}(X,\overline{W})$, respectively). We also refer to $I_{\theta
		c}(X,\overline{W})$ as the associated $I$-atom of $J_{\theta c}(X,\overline{W})$ and vice versa.
	
	\item
	We have also a query IDB predicate which is denoted $Q_1^{Datalog}(\overline{W})$
\end{enumerate}

Now, we describe the construction of the basic rules of the Datalog query which use the EDB predicates of the containing query and are as follows.
We call them basic because they do not depend on the ACs of the contained query.

\begin{enumerate}
	\item
	The {\em query rule} copies into its body all the relational subgoals of $Q_1$, and
	replaces each AC subgoal of $Q_1$ that compares a non-single-mapping variable to a constant by its associated $I$-atom.  The head of this rule is the same as the head of the query $Q_1$.
	\item
	We
	get one {\em mapping rule} for each SI arithmetic comparison  $e$ in $Q_1$ which is on a non-single-mapping variable.
	The body of each mapping rule is a copy of the body of the query rule, except that the
	$I$ atom associated with $e$ is deleted. The head is the  $J$ atom
	associated with $e$.
	\item
	For every pair of constants $c_1 \leq  c_2$ used in $Q_1$, we
construct three {\em coupling rules}.

	First, we construct the following two coupling rules: 
	\begin{center}
		\begin{tabular}{l}
			$I_{\leq c_2}(X,\overline{W})~\symif~J_{\geq c_1}(X,\overline{W})$ \\
			$I_{\geq c_1}(X,\overline{W})~\symif~J_{\leq c_2}(X,\overline{W})$
		\end{tabular}
	\end{center}

Then, we construct a coupling rule which is the following:
		$$I_{\leq c_2}(X,\overline{W})\symif~J_{\geq c_1}(Y,\overline{W}),U(X,Y).$$
%
%

\end{enumerate}



\begin{example}
	\label{ex-running-dat}
	For the query $Q_1$ of Example~\ref{ex:running41}, the construction we described yields the following basic rules of the Datalog
	query $Q_1^{Datalog}$:
	
	\begin{center}
		\begin{tabular}{l l l}
			$Q_1^{Datalog}(\overline{W})$ $\symif$       & $e(X,Y),e(Y,Z),a(\overline{W},Y),I_{\geq 5}(X,\overline{W}),$& \multirow{2}{*}{(query rule)}\\
			&$I_{\leq 8}(Z,\overline{W}).$& \\
			$J_{\leq 8}(Z,\overline{W})$ $\symif$ & $e(X,Y),e(Y,Z),a(\overline{W},Y),I_{\geq 5}(X,\overline{W}).$           & (mapping rule)\\
			$J_{\geq 5}(X,\overline{W})$ $\symif$ & $e(X,Y),e(Y,Z),a(\overline{W},Y),I_{\leq 8}(Z,\overline{W}).$           & (mapping rule)\\
			$I_{\leq 8}(X,\overline{W})$ $\symif$ & $J_{\geq 5}(X,\overline{W}).$                         & (coupling rule)\\
			$I_{\geq 5}(X,\overline{W})$ $\symif$ & $J_{\leq 8}(X,\overline{W}).$                         & (coupling rule)\\
			$I_{\leq 8}(X,\overline{W})$ $\symif$ &$J_{\geq 5}(Y,\overline{W}), U(X,Y)$ & (coupling rule)\\
			$I_{\geq 5}(X,\overline{W})$ $\symif$ &$J_{\leq 8}(Y,\overline{W}), U(Y,X)$ & (coupling rule)\\
		\end{tabular}
	\end{center}
	Intuitively, a coupling rule denotes that a formula $AC_1 \vee AC_2$ ( for two
	SI comparisons $AC_1=X\theta_1 c_1$ and $AC_2=Y\theta_2 c_2$) is either true or it is implied by $X\leq Y$ (which is encoded by  the predicate $U(X,Y)$). Thus, the first coupling rule in the above query says that
	$ X\leq 8  \vee X\geq 5$ is true and the second coupling rule says the same but refering to different
	$I$ and $J$-atoms. Moreover, the last coupling rule says that $X\leq Y\Rightarrow X\leq 8  \vee Y\geq 5$.
\end{example}
{\bf Construction of the dependant rules} $Q_1^{Datalog}:$
First, we describe the EDB predicates that we introduce (they all depend on the ACs of the contained query):

\begin{itemize}
	\item A unary predicate $U_{\theta c}(X,\overline{W})$, where $\theta$ is either $\leq$ or $\geq$ (the intuition for $\overline{W}$ is  that it will carry, during the computation, the head variables of the query rule), for each SI AC  $X\theta c$ in the closure of the ACs in the contained query. Note that although $U_{\theta c}$ typically includes $\overline{W}$, in the following, we could ignore it, for simplicity.
	
\end{itemize}

We have one kind of  dependant rules, the  \textit{link rules}:

\begin{itemize}
	\item
	For each pair of constants $(c_1,c_2)$, one in SIs of $Q_1$ and the other in an SI in the closure of ACs of $Q_2$ then, if $c_1\leq c_2$, we add the non-recursive link rule:
	$$ I_{\geq c_1}(X,\overline{W}):-U_{\geq c_2}(X,\overline{W}).$$
	Similarly, we do in a symmetric way for the $\leq$ ACs in $Q_1$ and $Q_2$.
	

\end{itemize}

Thus,
each link
rule encodes an
entailment of the form $X\leq 7 \imply X\leq 8$, i.e., it encodes,
in general, an entailment  $X\leq c_1 \imply X\leq c_2$ where $c_1
\leq c_2$. Intuitively, the link rules are used to link the ACs between the contained query and the containing query, as described through the containment entailment. Typically, the unary predicates represent the ACs of the contained query. 

For an example of dependant rules, see below (also analyzed in the next subsections):

	\begin{center}
	\begin{tabular}{l l l}
		$I_{\geq 5}(X,\overline{W})$ & $\symif~U_{\geq 6}(X,\overline{W}).$ & (link rule)\\
		$I_{\leq 8}(X,\overline{W})$ & $\symif~U_{\leq 7}(X,\overline{W}).$ & (link rule)\\
	\end{tabular}
\end{center}

\subsection{Construction of CQ for Contained Query}
\label{subsec-contained-trans}
%

We now describe the construction of the contained query turned into a CQ .

{\bf Construction of } $\mathbf{Q_2^{CQ}:}$ We introduce new unary
EDBs, specifically two of them, by the names $U_{\geq c}$ and $U_{\leq c}$, for each constant $c$ in $Q_2$. In addition, we use the binary predicate $U$ to represent the closed SI ACs between two variables, as we saw in the previous section.
Let us now construct the CQ $Q_2^{CQ}$ from $Q_2$.
We initially copy the regular subgoals of $Q_2$,
and for each SI  $X_i\theta c_i$ in the closure of $\beta_2$ we add a
unary predicate subgoal $U_{\theta c_i}(X_i)$. Then, for each AC $X\leq Y$ in the closure of ACs in $Q_2$, we add the unary subgoal $U(X,Y)$ in the body of the rule.

For example, considering the CQAC $Q_2$ with the following definition:

	\begin{center}
	\begin{tabular}{ll}
		$Q_2(W_1,W_2)\symif$ & $e(A,B),e(B,C),e(C,D),e(D,E), A\geq 6,$\\
		&$E\leq 7,
		a(W_1,W_2,B), a(W_1,W_2,D).$ \\
	\end{tabular}
\end{center}

we construct the $Q_2^{CQ}$ whose definition is:

	\begin{center}
	\begin{tabular}{ll}
		$Q_2^{CQ}(W_1,W_2)\symif$ & $e(A,B),e(B,C),e(C,D),e(D,E),U_{\geq 6}(A),$\\
		&$U_{\leq 7}(E),
		a(W_1,W_2,B), a(W_1,W_2,D).$ \\
	\end{tabular}
\end{center}

Thus the dependant rules for our running example, query $Q_1$, and the above contained query $Q_2$ are:

	\begin{center}
	\begin{tabular}{l l l}
		$I_{\geq 5}(X,\overline{W})$ & $\symif~U_{\geq 6}(X,\overline{W}).$ & (link rule)\\
		$I_{\leq 8}(X,\overline{W})$ & $\symif~U_{\leq 7}(X,\overline{W}).$ & (link rule)\\
	\end{tabular}
\end{center}


Now, we have completed the description of the construction of both $Q_1^{Datalog}$ from $Q_1$ and
$Q_2^{CQ}$ from $Q_2$. We go back to our examples and put all together.

\begin{example}
	\label{ex-full}
	Our contained query is the one in Subsection  \ref{subsec-contained-trans}.
	Our containing query is the one  in Example \ref{ex:running41}. The transformation of the contained query is shown in Subsection  \ref{subsec-contained-trans}.  The transformation of the contained query is shown in Example \ref{ex-running-dat}, where we see the basic rules. To complete the Datalog query, we add
	the following link rules:

	
	\begin{center}
		\begin{tabular}{l l l}
			$I_{\geq 5}(X,\overline{W})$ & $\symif~U_{\geq 6}(X,\overline{W}).$             &             (link rule)\\
			$I_{\leq 8}(X,\overline{W})$ & $\symif~U_{\leq 7}(X,\overline{W}).$              &            (link rule)\\
		\end{tabular}
	\end{center}
	In fact,  we constructed  the two new link rules in the Datalog query for $Q_1$. One rule links the constant 6 from the ACs of  $Q_2$ to the constant 5 from the ACs of  $Q_1$. The other link rule links constants 7 and 8 from queries $Q_1$ and $Q_2$, respectively.
\end{example}


%
%
%
\subsection{Proving the main theorem and the complexity}
\label{subsec-main-proof-sec5}

The constructions of the Datalog query and the CQ presented in Sections \ref{subsec-construct-Datalog} and \ref{subsec-contained-trans}, respectively, lead to the following theorem.

\begin{theorem}
	\label{thm:main}
	Consider two conjunctive queries with arithmetic comparisons, $Q_1$ and $Q_2$  such that ($Q_1,Q_2$) is an RSI1 disjoint-AC pair.
	%
	%
	%
	Then, $Q_1$ contains $Q_2$ if and only if  the following two happen a) $Q_1^{Datalog}$ contains  $Q_2^{CQ}$ and b) the head entailment is true.
\end{theorem}


The challenging part of the Theorem~\ref{thm:main} concerns the part (a) which is restated in the Theorem~\ref{thm:main123}. The part (b) of Theorem~\ref{thm:main} is a straightforward consequence of Proposition~\ref{prop:single-map-vars}.

\begin{theorem}
	\label{thm:main123}
	Consider two conjunctive queries with arithmetic comparisons, $Q_1$ and $Q_2$  such that ($Q_1,Q_2$) is an RSI1 disjoint-AC pair.
	Let  $Q_1^{Datalog}$ be the  transformed Datalog query of $Q_ 1$. Let
	$Q_2^{CQ}$ be the  transformed CQ query of $Q_2$.
	Then, the  body containment entailment for containment of $Q_2$ to $Q_1$  is true  if and only if  $Q_1^{Datalog}$ contains  $Q_2^{CQ}$.
\end{theorem}

The proof of Theorem \ref{thm:main123} is in the \ref{prf:thm-main123p}. The following theorem proves that checking body containment entailment is NP-complete.

\begin{theorem}
	\label{thm:datalog-np-complete}
	Consider two conjunctive queries with arithmetic comparisons, $Q_1$ and $Q_2$  such that ($Q_1,Q_2$) is an RSI1 disjoint-AC pair. Let  $Q_1^{Datalog}$ be the  transformed Datalog query of $Q_ 1$. Let
	$Q_2^{CQ}$ be the  transformed CQ query of $Q_2$. Checking whether $Q_2^{CQ}$ is contained in $Q_1^{Datalog}$ is NP-complete.
\end{theorem}

Theorem~\ref{thm:datalog-np-complete} can be generalized to a stronger result, which is presented in Section~\ref{sec:semi-monadic-datalog-cont} in Theorem~\ref{thm-semi-monadic-np}. Theorem \ref{thm-np-complete}  is a straightforward consequence of Theorem~\ref{thm:check-head-body-entailments}.

\begin{theorem}
	\label{thm:check-head-body-entailments}
	Consider two conjunctive queries with arithmetic comparisons, $Q_1$ and $Q_2$  such that ($Q_1,Q_2$) is an RSI1 disjoint-AC pair.  Let $\phi_h$ and $\phi_b$ be the head and body entailments, respectively. Then, checking $\phi_h$ is polynomial and checking $\phi_b$ is NP-complete.
	
\end{theorem}

To prove that checking $\phi_h$ is polynomial, observe that it suffices to compute the closure of a set of ACs. This can be done in polynomial time.

%

Consider two conjunctive queries with arithmetic comparisons, $Q_1$ and $Q_2$  such that $Q_1$ is an RSI1+ query and $Q_2$ is a CQAC with closed ACs.  It is straightforward that ($Q_1$, $Q_2$) is a RSI1 disjoint-AC pair with respect to the set of head variables of $Q_1$.

The following is a corollary of Theorem \ref{thm:check-head-body-entailments}.
\begin{corollary}
	\label{cor:check-head-body-entailments}
	
	Consider two conjunctive queries with arithmetic comparisons, $Q_1$ and $Q_2$  such that $Q_1$ is an RSI1+ query and $Q_2$ is a CQAC with closed ACs.  
	Let $\phi_h$ and $\phi_b$ be the head and body entailments, respectively. Then, checking $\phi_h$ is polynomial and checking $\phi_b$ is NP-complete.
	
\end{corollary}

\subsection{More examples to illustrate the technique}
\label{subsec-examples}

Another example to use later to illustrate the functionality of the second kind of coupling rules.

%
%
%
%
%
%
%

\begin{example}
	\label{ex:running412}
	Consider a relational schema with the binary relations $e$ and $a$, as well as the following two CQACs over this schema.
	
	\begin{center}
		\begin{tabular}{lll}
			$Q_1: q(W_1,W_2)$ & $\symif$ &$a(W_1,W_2,Y),e(X,Y),e(Y,Z),X\geq 5,Z\leq 5$ \\
			$Q_2: q(W_1,W_2)$ & $\symif$ &$e(A,B),e(B,C_1),e(C_2,D),e(D,E),a(W_1,W_2,B),$\\
			& &$ a(W_1,W_2,D), C_1\leq C_2,A\geq 5,E\leq 5$\\
		\end{tabular}
	\end{center}
	
	Checking the containment $Q_2\sqsubseteq Q_1$, note that there are two containment mappings $\mu_1$, $\mu_2$ from $Q_{10}$ to $Q_{20}$ such that $\mu_1(W_i)=\mu_2(W_i)=W_i$, and
	\begin{itemize}
		\item $\mu_1:$ $Y\rightarrow B$, $X\rightarrow A$, $Z\rightarrow C_1$.
		\item $\mu_2:$ $Y\rightarrow D$, $X\rightarrow C_2$, $Z\rightarrow E$.
	\end{itemize}
	Then, applying the mappings on the query entailment we conclude the following implication:
	$$((C_1\leq C_2)\wedge(A\geq 5)\wedge(E\leq 5))\Rightarrow((A\geq 5)\wedge (C_1\leq 5))  \vee ((C_2\geq 5)\wedge (E\leq 5))$$
	
	
	Analyzing the aforementioned entailment, it is easy to verify that it is true, since $(C_1\leq C_2)\Rightarrow (C_1\leq c)\vee (C_2\geq c)$ is true for every constant $c$; hence, $Q_2\sqsubseteq Q_1$.
	
	Let us now construct $Q_1^{Datalog}$ from $Q_1$ and $Q_2^{CQ}$ from $Q_2$.  To construct $Q_1^{Datalog}$ from $Q_1$ we follow the algorithm in Section~\ref{subsec-construct-Datalog}. In particular, we initially construct the query rule, which is given as follows. For simplicity in the notation, we will denote by $\overline{W}$  the vector of head variables $W_1,W_2$. Note that the subgoals $I_{\geq 5}(X,\overline{W})$, $I_{\leq 5}(Z,\overline{W})$ correspond to the ACs $X\geq 5$ and $Z\leq 5$, respectively.
	
	\begin{center}
		\begin{tabular}{ll}
			$Q_1^{Datalog}: q(\overline{W})$ $\symif$ &$e(X,Y),e(Y,Z),a(\overline{W},Y),I_{\geq 5}(X,\overline{W}),I_{\leq 5}(Z,\overline{W})$\\
		\end{tabular}
	\end{center}
	
	Then, we construct the basic mapping and coupling rules, which are given by the following rules:
	
	\begin{center}
		\begin{tabular}{lll}
			$J_{\geq 5}(X,\overline{W})$ $\symif$ &$e(X,Y),e(Y,Z),a(\overline{W},Y),I_{\leq 5}(Z,\overline{W})$ & (mapping rule)\\
			$J_{\leq 5}(Z,\overline{W})$ $\symif$ &$e(X,Y),e(Y,Z),a(\overline{W},Y),I_{\geq 5}(X,\overline{W})$ & (mapping rule)\\
			$I_{\leq 5}(X,\overline{W})$ $\symif$ &$J_{\geq 5}(X,\overline{W})$ & (coupling rule)\\
			$I_{\geq 5}(X,\overline{W})$ $\symif$ &$J_{\leq 5}(X,\overline{W})$ & (coupling rule)\\
			$I_{\leq 5}(X,\overline{W})$ $\symif$ &$J_{\geq 5}(Y,\overline{W}), U(X,Y)$ & (coupling rule)\\
			$I_{\geq 5}(X,\overline{W})$ $\symif$ &$J_{\leq 5}(Y,\overline{W}), U(X,Y)$ & (coupling rule)\\
		\end{tabular}
	\end{center}
	
	To find the $Q_2^{CQ}$, we initially copy the head $Q_2$, along with its relational subgoals. Then, we consider the subgoal $U(C_1,C_2)$ representing the AC $C_1\leq C_2$, as well as the unary suboals $U_{\geq 5}(A,\overline{W})$ and $U_{\leq 5}(E,\overline{W})$ to represent the ACs $A\geq 5$ and $E\leq 5$, respectively. Consequently, we end up with the following CQ definition:
	
	\begin{center}
		\begin{tabular}{lll}
			$Q_2^{CQ}: q(W_1,W_2)$ & $\symif$ &$e(A,B),e(B,C_1),e(C_2,D),e(D,E),a(W_1,W_2,B),$\\
			& &$ a(W_1,W_2,D), U(C_1,C_2),U_{\geq 5}(A,\overline{W}),U_{\leq 5}(E,\overline{W})$\\
		\end{tabular}
	\end{center}
	
	Finally, the link rules included in the Datalog query $Q_1^{Datalog}$ are constructed as follows:
	\begin{center}
		\begin{tabular}{ll}
			$I_{\leq 5}(X,\overline{W})$ $\symif$ &$U_{\leq 5}(X,\overline{W})$\\
			$I_{\geq 5}(X,\overline{W})$ $\symif$ &$U_{\geq 5}(X,\overline{W})$\\
		\end{tabular}
	\end{center}
	
	%
	
	%
	%
	%
	%
\end{example}

{\sl Useful observation:}
Notice that, because of the restrictions we have assumed on our queries, $\overline{W}$ as it appears in the construction of the Datalog query does not contain any of the variables in the first position of a semi-unary predicate.

Finally, it helps with the inuition to obseerve the following: Even if the query $Q_1$ was different but only as concerns AC that involve head variables, the Datalog query would be the same because we do the test for such ACs in the preliminary step. Thus the following CQAC would have been transformed to the same query as above:
%
%

\begin{center}
	\begin{tabular}{ll}
		$Q_1({W_1,W_2})\symif$ & $a(W_1,W_2,Y), e(X,Y),e(Y,Z),$\\
		& $X\geq 5,Z\leq 8,W_1<W_2,W_1< 4.$\\
	\end{tabular}
\end{center}

%
%

\subsection{Preliminary partial results and intuition on the proof of Theorem \ref{thm:main123}  
}
\label{subsec-simplefacts}

The proof of Theorem \ref{thm:main123} is  presented in the \ref{prf:thm-main123p}. Here we give some insight into the technicalities involved in its proof.

In our proof, we will apply the Datalog query  $Q_1^{Datalog}$ on the canonical database of the CQ query $Q_2^{CQ}$ constructed from the contained query $Q_2$.
This canonical database uses constants (different from the constants in the ACs) that correspond one-to-one to variables of the query $Q_2$.
Thus, as we compute facts, each fact being either an $I$ fact or a $J$ fact, we do the following observations about the result of firings for each of the  two kinds of recursive rules (i.e., the coupling rules and the mapping rules):
(all the $\theta$s represent either $\leq$ or $\geq$ and the $c_i$s are constants from the ACs of the queries.
\begin{itemize}
	\item We have two kinds of coupling rules. Consider a coupling rule of the first kind which is of the form: $$I_{\theta_1 c_1}(X,\overline{W}) \symif~J_{\theta c_2}(X,\overline{W}).$$
	When this rule is fired, its variable $X$ is instantiated to a constant, $y$, in the canonical database, $D$, of $Q_{20}$. The constant $y$ corresponds to the variable $Y$ of $Q_2$ by convention. Then the following is true by construction: $X\theta_1 c_1 \vee X\theta_2 c_2$, and, hence, the following is true:
	$$\beta_2\Rightarrow X\theta_1 c_1 \vee X\theta_2 c_2$$
	Now consider the other kind of coupling rule, which is of the form:
	$$I_{\theta_1 c_1}(X,\overline{W}) \symif~J_{\theta c_2}(Y,\overline{W}),U(X,Y).$$
	By construction of the rule, the EDB $U(X,Y)$    is mapped in $D$ to two constants/variables 
	such that there in $Q_2$ an AC which is $X\leq Y$.       Thus, by construction of the rule, the following is true again:
	$$\beta_2\Rightarrow X\theta_1 c_1 \vee Y\theta_2 c_2$$
	We say in both cases of coupling rules that the facts in both sides of the rule are coupled and that the corresponding ACs are coupled.
	\item Consider a mapping rule
	$$J_{\theta_1 c_1}(Z,\overline{W}) \symif~BodyQ_1,I_{\theta_2 c_2}(X,\overline{W}), I_{\theta_3 c_3}(X,\overline{W}),\dots.$$
	The $BodyQ_1$ denotes all the relational subgoals of $Q_1$. When a mapping rule is fired, then there is a containment mapping, $\mu$,  from the relational subgoals of $Q_1$ to the relational subgoals of $Q_2$ and, moreover, the $I$ facts in the body of the rule have been computed in previous rounds of the computation.
	
	The $I$ facts can be computed either via link rules or via coupling rules. When the $I$ facts in the body of the rule (for the instantiation that fires the rule) are  computed via coupling rules using $J$ facts, each  $I$ fact is coupled with a $J$ fact.
	Notice that each $I$ fact corresponds to an AC in $\mu(\beta_1)$ by construction of a mapping rule.
	Putting the implications we derived for coupling rules above together for all $I$ facts in the body of the mapping rule,
	we derive the implication:
	$$\beta_2 \Rightarrow \mu (\beta_1)\vee e_1 \vee e_2 \vee \cdots , \vee e_t$$
	where $ e_1 , e_2, \ldots$ are the ACs corresponding to the $J$ facts from which each $I$ fact was computed. Finally, observe that by construction of the rule, one of the ACs in $\mu(\beta_1)$ is not represented in the body of the rule (it is represented in the head of the rule). This justifies the
	presence of $e_t$ in the implication, which represents this special AC in $\mu(\beta_1)$.
\end{itemize}

\section{When U-CQAC MCRs compute certain answers}
\label{sec-6}

\def\L{\mathcal{L}}
\def\I{\mathcal{I}}
\def\P{\mathcal{P}}
\def\D{\mathcal{D}}
\def\V{\mathcal{V}}
\def\C{\mathcal{C}}
\def\T{\mathcal{T}}
\def\S{\mathcal{S}}
\def\U{\mathcal{U}}
\def\M{\mathcal{M}}
\def\E{\mathcal{E}}
\def\K{\mathcal{K}}
\def\LAV{\textsf{LAV}}
\def\DeptA{\textsf{Dept}_\textsf{A}}
\def\DeptB{\textsf{Dept}_\textsf{B}}
\def\Staff{\textsf{Staff}}
\def\corecoverc{\textsf{CoreCover$\C$}}

In this section we prove that, given CQAC query and views, if there is a maximally contained rewriting (MCR) in the language of (possibly infinite) union of CQACs then this MCR computes all the certain answers on any view instance $\I$. This section extends the results in \cite{AfratiK10} for CQs.

Moreover, we prove this result in a more general setting, in that we also assume that there is a set of constraints $\C$ that the database ought to satisfy.
The set $\C$ contains tuple generating dependencies (tgds) and equality generating dependencies (egds). We assume that the chase algorithm (see description of chase algorithm as well as definitions for tgds and egds  in \ref{pre:dependencies-chase}) terminates on $\C$.

We give the definition of certain answers under constraints, as follows.

\begin{definition}
	Suppose there exists a database instance $D$ such that $\I\subseteq \V(D)$. Then,
	we define the certain answers of ($Q,\I$) with respect to $\V$ as follows:
	\begin{itemize}
		\item Under the Open World Assumption:
		\[\text{certain}(Q,\I)=\bigcap\{Q(D): D \text{ such that } \I\subseteq\V(D)\}\] In the presence of a set of constraints $\C$, we also require that all databases $D$ used for certain$(Q,\I)$ satisfy $\C$ and denote it by  $\text{certain}_{\C}(Q,\I)$.
	\end{itemize}
	
	If there is no database instance $D$ such that $\I\subseteq \V(D)$, we say that the set $\text{certain}_{\C}(Q,\I)$ is undefined. 
\end{definition}

\subsection{Preliminaries}

We first  define query containment  under constraints:
\begin{definition}
	Let $\C$ be a set of tdgs and egds, and $Q_1$, $Q_2$ be two conjunctive queries. We say that {\em $Q_1$ is contained in $Q_2$ under the dependencies $\C$,} denoted $Q_1\sqsubseteq_\C Q_2$, if for all databases $D$ that satisfy $\C$ we have that $Q_1(D)\subseteq Q_2(D)$.
\end{definition}

We check  CQAC containment under contstraints $\C$  by using the $\C$-canonical databases (see \ref{pre:dependencies-chase1}).
We define contained rewriting under constraints:

\begin{definition}
	\label{cont-rewr-dfn}
	(Contained rewriting)
	Let $Q$ be a query defined on schema $\cal S$, and $\cal{V}$ a set of views defined on $\cal S$. Let $R$ be a query formulated in terms of the view relations in
	the set $\cal{V}$.
	
	$R$ is a {\em contained rewriting of $Q$ using $\V$ under the OWA and under the constraints $\C$} if and only if for every view instance $\I$ the following is true:
	For any database $D$ such that $\I\subseteq \V(D)$ that satisfies the constraints in $\C$, we have that $R(\I)\subseteq Q(D)$.
\end{definition}


\begin{theorem}
	Suppose query $Q$, views $\cal{V}$, and rewriting $R$ all belong to the language of CQACs. Then $R$ is a contained rewriting of
	$Q$ using views  $\cal{V}$ if and only if $R^{exp}\sqsubseteq_\C  Q$.
\end{theorem}

%

\begin{proof} If the expansion is not contained in the query, then we find a counterexample to prove that it is not a contained rewriting as follows: Since $R^{exp}$ is not contained in $Q$, there is a $\C$-canonical database $D$ of $R^{exp}$ such that a tuple $t$ is computed by $Q_2$ on $D$ but not by $Q_1$. We compute $\V$ on $D$ and produce view instance $\I$. Then $t$ is in $R(\I)$ (because a subset of $\I$ is isomorphic to the body of $R$) but $t$ is not in $Q(D)$.

	If the expansion $R^{exp}$ is contained in the query then, since $\I \subseteq \V(D)$ for any $D$ that satisfies the constraints, we have that
	$ R(\I) \subseteq R(\V(D) ) $.
	However $R(\V(D) )$ is equal to $R^{exp}(D )$ because to compute the former we first apply the mappings from the view definition to $D$ (to compute $\V(D)$) and then apply the mapping from $R$ to $\V(D)$ thus resulting in a mapping from $R^{exp}$ to $D$ for each tuple that is computed.
	Consequently, the following is true:
	$$ R(\I) \subseteq R(\V(D) ) \subseteq  R^{exp}(D ) \subseteq Q(D)$$
	for any $D$ that satisfies the constraints. Hence $R$ is a contained rewriting under the constraints.
\end{proof}
%
%


\subsubsection{Database AC-instance with t-instance}

A {\em database AC-instance with ACs} is a database with domain a set of constants and a set of variables that we call {\em labeled nulls} (the two sets are disjoint), i.e., it contains  relational atoms that use  labeled nulls and constants. It may also contain ACs among the labeled nulls or among labeled nulls and constants. 
When the ACs define a total ordering, then we call $\I$ a \textit{t-instance}.

%


%

 Let $J_1$, $J_2$ be sets of atoms over the schema $\CS$ such that $J_1$ is an AC-instance  and $J_2$ is a t-instance. An 
\textit{order-homomorphism}p $h : J_1\rightarrow J_2$ is a mapping from the atoms in $J_1$ to the atoms in $J_2$ with the following properties:

\begin{enumerate}

	\item For every constant $c$ in $J_1$, we have $h(c) = c$.

	\item For every atom $r(X_1,\dots,X_m)$ in $J_1$, we have that $r(h(X_1),\dots,h(X_m))$ is an atom in $J_2$, where $X_1,\dots,X_m$ are either variables or constants.

	\item if $(X_1\;\theta\;X_2)$ is true in $J_1$, where $\theta$ is $<, >, \leq, \geq, =$, then $(h(X_1)\;\theta\;h(X_2))$ is implied by the partial order of $J_2$.


\end{enumerate}


\subsection{Representative possible worlds (RPW)}

In this section, we will prove that, for  CQAC views, a
maximally contained rewriting $\P$  with respect to U-CQAC  \footnote{In the literature, usually, by U-CQAC we define the class of finite unions of CQACs, in this section we assume that it may be also infinite.} of a CQAC query $Q$ under a given set of constraints computes the certain answers of $Q$ under the OWA, i.e., we prove the following theorem.

\begin{theorem}\label{owa_mcrs}
	Let $\C$ be a set of constraints that are tgds and egds.
	Let $Q$ be a CQAC query, $\V$ a set of CQAC views. Suppose there exists an MCR $\R_{MCR}$ of $Q$ with respect to U-CQAC and under the constraints $\C$. Let $\I$ be a view instance such that the set $\text{certain}_{\C}(Q,\I)$ is defined. 
	Then, under the open world assumption, $\R_{MCR}$ computes all the certain answers of $Q$ on any view instance $\I$ under the constraints $\C$, that is: $\R_{MCR}(\I)=\text{certain}_{\C}(Q,\I)$.
\end{theorem}

We define the concept of representative possible worlds of a view instance $\I$ in order to analyze how we compute the certain answers.

Given a view instance $\I$,  we define a set of {\em representative possible worlds (RPW, for short)} $\P_\I$. A RPW is a AC-instance. The set  $\P_\I$ has the following properties: a) for all
$D_\I\subseteq \P_\I$ the following is true:  $\I\subseteq\V(D_\I)$, b) for each database instance $D$ such that $\I\subseteq\V(D)$ there is a representative possible world $D_{\I}$ in $\P_\I$ such that there is an 
order-homomorphism
from  $D_{\I}$ to $D$


The set  $\P_\I$ of RPWs is finite and we can construct it by the following algorithm, consisting of two main stages:

\paragraph{Stage 1:} In this stage we construct a Boolean query. 
Let $\I$ be a view instance. We use $\I$ to produce a Boolean CQAC rewriting, $R_{\I}$, as follows:\footnote{A rewriting is a CQAC query expressed in terms of the views; it stands alone, it does not have to be contained in a specific query.}

\begin{enumerate}
	\item We turn all the constants in $\I$ to variables so that distinct constants are turned into distinct variables.
	\item We add on the variables the ACs that imply a total ordering, which is the ordering of the constants they came from (recall that constants are from a totally ordered domain).
\end{enumerate}

\paragraph{Stage 2:} The following steps construct the set of RPWs:
\begin{enumerate}
	\item We consider the expansion $R_{\I}^{exp}$ of $R_{\I}$. We consider the set ${\cal{R}}_{\I}$ of the canonical databases of $R_{\I}^{exp}$ for which $R_{\I}^{exp}$ computes to true.
	Each element of $R_{\I}$ is a database t-instance.
	\item For each $D$ in ${\cal{R}}_{\I}$, we do as follows: We apply the chase on $D$ with constraints 
	$\C$. Thus, if the chase succeeds, we derive $D_{chased}$ and
	add it in $\P_\I$ which is the set of representative possible worlds.
\end{enumerate}

This finishes the construction of $\P_\I$. Notice that the databases in $\P_\I$ are exactly all the $\C$-canonical databases of $R_{\I}^{exp}$.

\begin{theorem}
	The above procedure finds all representative possible worlds of the view instance $\I$.
\end{theorem}

\begin{proof}
	Let $D$ be a database instance that satisfies the constraints $\C$ and such that $\I\subseteq \V(D)$. The tuples in $\I \cap \V(D)$ are produced by an
	order-homomorphism,
	 $ h_1$, from  $R_{\I}^{exp}$ to $D$. To see that, imagine that we apply the view definitions in $\V$ on $D$ in one step (since we know that $\I\subseteq \V(D)$).
	
	This means that $D$ is contained under $\C$ (we can imagine that $D$ is a Boolean query with no variables, just constants) in $R_{\I}^{exp}$.
	Thus, by the containment test, and taking into account  Theorem~\ref{thm-chase-main22},  there is a $\C$-canonical database of $R_{\I}^{exp}$ that maps isomorphically on $D$ by $h_2$ according to the following proposition. 

	
	\begin{proposition}
		\label{pro-RPW-property}
		Suppose database instance $D$ which, viewed as a Boolean query, is contained in a CQAC $Q$. Then, there is a canonical database of $Q$ that maps isomorphically on $D$.  
	\end{proposition}
\begin{proof}
The proof of this proposition results from the observation that, by definition, the canonical databases of $Q$ represent all homomorphic images of the relational atoms of $Q$ that satisfy the ACs in $Q$. 
\end{proof}
	\vspace*{-.6cm}
\end{proof}


The following is an example showing  how we construct $R_{\I}$ and  $R_{\I}^{exp}$.
\begin{example}
	Consider the query $Q$ and the views $V_1$, $V_2$ with the following definitions.
	
	\begin{center}
		\begin{tabular}{ll}
			$Q: q()$ $\symif$ &$a(X,Y,W), b(Y,Z,W), X\leq 14$\\
		\end{tabular}
	\end{center}
	\begin{center}
		\begin{tabular}{ll}
			$V_1: v_1(X,Y)$ $\symif$ &$a(X,Y,Z),X\leq 9$\\
			$V_2: v_2(X,Y)$ $\symif$ &$b(X,Y,Z)$\\
		\end{tabular}
	\end{center}
	Now, we consider the following view instance:
	$\I:\{v_1(1,2),v_2(2,3),v_1(5,6)\}$.
	We build a Boolean rewriting from $\I$ as we explained above, which, in this specific view instance is the following rewriting:
	
	\begin{center}
		\begin{tabular}{ll}
			$R_{\I}: q()$ $\symif$ &$v_1(X_1,X_2), v_2(X_2,X_3), v_1(X_5,X_6),$\\
			&$X_1< X_2,X_2< X_3,X_5< X_6,X_3< X_5$\\
		\end{tabular}
	\end{center}
where, variable $X_1$ represents constant 1, variable $X_2$ represents constant 2, etc. Since $1<2<3<5<6$, we have added in the above query $X_1< X_2,X_2< X_3,X_5< X_6,X_3< X_5$.
	
	This rewriting is a contained rewriting in the query $Q$. However this is not always the case, e.g., imagine a view instance $\I'$ that contained only $v_1(5,6)$; it is easy to verify that the rewriting built based on this view instance $\I'$ would not have been contained in $Q$.
	
	The expansion of the rewriting $R_{\I}$ is the following:
	
	\begin{center}
		\begin{tabular}{ll}
			$R^{exp}_{\I}: q()$ $\symif$ &$a(X_1,X_2,Z_1), b(X_2,X_3,Z_2), a(X_5,X_6,Z_3),$\\
			&$X_1< X_2,X_2< X_3,X_5< X_6,X_3< X_5$\\
		\end{tabular}
	\end{center}
	
	The representative possibe worlds for $\I:\{v_1(1,2),v_2(2,3),v_1(5,6)\}$ are obtained  from the canonical databases of the expansion $R^{exp}_{\I}$. Each RPW contains the relational atoms in  $R^{exp}_{\I}$ and the variables (labeled nulls) $X_i$ have the total order shown in  $R^{exp}_{\I}$. However the variables (labeled nulls) $Z_i$ can have any ordering, thus all their orderings create more than one RPW. 
\end{example}

\subsection{When a view instance has at least one representative possible world}

There is a broad class of views where the set $\text{certain}_{\C}(Q,\I)$ is always defined  independently of the view instance $\I$, as the following proposition shows. 
\begin{proposition}\label{repeat_vars_lemma}
	Let $\V$ be a set of CQAC views and $Q$ a CQAC query. If there are no egds in the set of constraints $\C$ and, each view definition a) has no repeated variables in the head and b) has no ACs that contain head variables, then the set $\text{certain}_{\C}(Q,\I)$ is defined on any view instance $\I$.
	%
	%
\end{proposition}
\begin{proof}
	When we construct the RPWs, for each view tuple in the view instance $\I$, we associate position-wise each variable in the head of the view definition with a constant in the view tuple. This should create an order-homomorphism from the head of the view definition to the view tuple. This is possible because there are no duplicate variables and no ACs on the head variables that could be violated.
\vspace*{-.5cm}
\end{proof}
\vspace*{-.4cm}
Towards future work, we begin a discussion on  it in Section \ref{sec-app-unclean} to argue that 
even in the case where certain answers are not defined, an MCR can be used to produce results that ``make sense'', when we assume that we are dealing with 
non-clean data.

\subsection{Main result  }

We will now prove Theorem \ref{owa_mcrs}, which is the main result of this section and its main ingredients are the following Propositions \ref{pro-qi-certain} and \ref{owa_mcrs_lemma}. The first says that if we take  the intersection of all the answers computed by applying the query $Q$ on each of the representative possible worlds we produce all the certain answers of the query. The second one says that there is a CQAC contained rewriting that produces this intersection.
We also need to use the fact that  each CQAC contained rewriting computes only certain answers if applied on a view instance $\I$; this is true by the definition of contained rewriting (Definition \ref{cont-rewr-dfn}).

\begin{proposition}
	\label{pro-qi-certain}
	Let  $\C$ be a set of constraints that are tgds and egds. Let $\V$ be a set of CQAC views and $\I$ a view instance such that the set  $\text{certain}_{\C}(Q,\I)$ is defined.
	Let $Q$ be a CQAC query.
	Then $\bigcap_{\D_\I \in \P_{\I}} Q(D_\I)$ is equal to the certain answers of $Q$ given $\V$ on view instance $\I$ under the constraints $\C$, where $ \P_{\I}$ is the set of representative possible worlds on $\I$.
\end{proposition}
\begin{proof}
	Certainly, $\bigcap_{\D_\I \in \P_{\I}} Q(D_\I)$ is a superset of the set of certain answers. 
	We want to prove that it is also a subset of the set of certain answers.  By contradiction, suppose not. Then, there is a PW $D$  such that the answers of $Q$ on $D$ do not contain all the tuples in $\bigcap_{\D_\I \in \P_{\I}} Q(D_\I)$. This means that there is a tuple $t$ in $\bigcap_{\D_\I \in \P_{\I}} Q(D_\I)$
	which is not in $Q(D)$. However, according to the definition of RPW, there is a RPW
 $D_r$ such that there is an 
	order-homomorphism
	from $D_r$ to $D$, hence $Q(D_r)\subseteq Q(D)$. Since $t$ is in $\bigcap_{\D_\I \in \P_{\I}} Q(D_\I)$, $t$ is also in $Q(D_r)$. Hence contradiction. 
	%
	%
\end{proof}

\begin{proposition}\label{owa_mcrs_lemma}
	Let  $\C$ be a set of constraints that are tgds and egds.
	Let $Q$ be CQAC query and $\V$ be a set of CQAC views. Let $\I$ be a view instance such that
	the set $\text{certain}_{\C}(Q,\I)$ is defined.
	Then, given a
	tuple $t_0 \in certain(Q,\I)$,  there is a contained CQAC rewriting $R$ 
	such that $t_0 \in R(\I)$.
\end{proposition}

\begin{proof}

	%
	%
	%
	%
	
	We consider as $R$ the Boolean query $R_{\I}$ with the proper variables in the head that are the variables that represent the constants in $t_0$.
	
	Now we need to prove that $R$ is a contained rewriting. $R$ was created from $R_{\I}$ which produces all the RPWs. Since $t_0$ is in the certain answers of the query $Q$, there is a 
	order-homomorphism
from $Q$ to every RPW and this order-homomorphism produces $t_0$. 
	All the RPWs are all the canonical databases of $R_{\I}^{exp}$ chased with the constraints.  Hence the previously mentioned order-homomorphisms provide the proof for the containment test that proves containment of $R_{\I}^{exp}$ to  $Q$ under the constraints $\C$. Since $R$ only differs from $R_{\I}$  as to the head, the same 
	order-homomorphisms
	can be used to prove containment of $R^{exp}$ to  $Q$ under the constraints $\C$. 
	%
	%
	%
	%
\end{proof}
%

We now put all together to finish the   proof of Theorem \ref{owa_mcrs}:
\begin{proof}{\bf (Theorem \ref{owa_mcrs})}
	We will show the following:
	\begin{enumerate}
		\item $\P(\I) \subseteq$ certain$(Q,\I)$
		\item certain$(Q,\I)\subseteq \P(\I)$
	\end{enumerate}
	Since $\P$ is a contained rewriting of $Q$, the first is a direct consequence of the definition of a contained rewriting.
	
	To prove (2), we use the two propositions. One proposition says that we can compute all the certain answers by considering only a finite number of possible worlds, $\P(\I)$.  The other one uses $\P(\I)$ to prove that there is   CQAC contained rewriting which computes a tuple $t_0$ if this tuple is in certain answers.
\end{proof}


\section{Finding MCR for CQAC-RSI1+ Query and CQAC$^-$ Views}
\label{subsec-mcr-datalogAC}

In this section, we show  that for a RSI1+ query and a special case of CQAC views, we can find an MCR in the language of (possibly infinite)  union of CQACs.
We will show that this MCR is expressed in
Datalog$^{AC}$. In detail, we consider
the following case of query and views:
\begin{itemize}
	\item There are only closed  arithmetic comparisons in both query and views.
	\item The views are  CQAC queries which do not use ACs of the form $X\leq Y$ or $X\geq Y$ where $X$ is a head variable and $Y$ is a nondistinguished variable. We call this class of CQAC queries CQAC$^{-}$.
	\item The query is a esRSI1+ query.
	%
\end{itemize}


We think of an expansion of a rewriting as having three kinds of variables: a) the head variables, b) the {\em view-head} variables, which   are all the variables that are present in the rewriting and c) the {\em view-nondistinguished} variables, which are  all the other variables in the expansion of the rewriting (these do not appear in the rewriting). The head variables are also view-head variables.
\subsection{ACs in rewritings}
\label{subsec-rectified}
Consider  a CQAC query and a set of CQAC views.
When we have a  rewriting $R$ the  variables in the rewriting also satisfy some ACs that are in the closure of the ACs in the expansion of the rewriting. We include those ACs in the rewriting $R$ and produce $R'$, which we call the {\em AC-rectified rewriting of $R$}. Thus, the expansions of $R$ and $R'$ are equivalent queries. Hence, we derive the following proposition:

\begin{proposition}
	Given  a set of CQAC views, a  rewriting $R$  and its rectified version $R'$, the following is true: For any
	view instance $\I$ such that there is a database instance $D$  for which  $\I\subseteq \V(D)$, we have that $R(\I)=R'(\I)$.
\end{proposition}

\begin{definition}
	We say that a  rewriting $R$ is {\em AC-contained} in a   rewriting $R_1$  if the AC-rectified rewriting $R'$ of $R$ is contained in $R_1$ as queries.
\end{definition}
From hereon,
when we refer to a  rewriting, we mean the AC-rectified version of it and when we say that a   rewriting is contained in another  rewriting we mean that it is AC-contained.
An example follows.
\begin{example}
	\label{ex:export-nondist1}
	Consider query $Q$ and view $V_2$:
	\begin{center}
		\begin{tabular}{l l}
			$Q(A)$     & $~\hbox{\rm :-}~~p(A), A < 4.$\\
			$V_2(Y,Z)$ & $~\hbox{\rm :-}~~p(X), s(Y,Z), Y \leq X, X \leq Z.$\\
		\end{tabular}
	\end{center}
	The following rewriting is a contained rewriting of the query in terms of the view in the language CQAC:
	$$R(Y_1)~~\hbox{\rm :-}~~V_2(Y_1,Z_1),V_2(Y_2,Z_2), Z_1\leq Y_2, Y_1\geq Z_2, Y_1 < 4.$$
	
	\noindent
	Now consider the following contained rewriting:
	
	$R'(X) $ ~\hbox{\rm :-}~ $V_2(X,X), X < 4.$
	
	\noindent
	This rewriting uses only one copy of the view. We can show that $R$ is not contained in $R'$ and that
	$R'$ is not contained in $R$. However, they compute the same output on any view instance (to see that just include in $R$ the ACs $Y_1\leq Z_1$ and $Y_2\leq Z_2$).
	
	
\end{example}

\subsection{Building MCRs for RSI1 queries}
\label{subsec-buildingMCRs}
First, in this subsection, we present the algorithm for building an MCR in the language of (possibly infinite) union of CQACs for the case of CQAC$^-$ views and queries that are RSI1. The algorithm for building  an MCR  for query $Q$ and view set $\V$ is the following:

%
%
%
%

{\bf Algorithm MCR-RSI1}:

\begin{enumerate}
	
	\item For the query $Q$, we construct the Datalog query $Q^{Datalog}$.
	We use the construction in Subsection \ref{subsec-construct-Datalog}. The link rules will use the constants present in the views and in the query.
	
	\item For each view $v_i$ in $\V$, we construct a new view $v_i^{CQ}$.  We use the construction in Subsection \ref{subsec-contained-trans}.
	
	\item Consider the EDB predicates introduced  in Section  \ref{ssec:mot} (and used in Steps 2 and 3 above) which encode ACs. We call them AC-EDB predicates and use them to construct a new set of {\em auxiliary} views as follows: a) Views with head  $u_{\theta c}$, one for each
	semi-unary predicate $U_{\theta c}$. The definition is $u_{\theta c}(\overline{W},X)~\symif~
	U_{\theta c}(\overline{W},X)$. b) A single view  $u$, whose definition is $u(X,Y)~\symif~
	U(X,Y)$. We will refer to those EDB predicates (i.e., the $U(X,Y)$ predicate and the semi-unary predicates) as {\em AC-predicates} or {\em AC-subgoals}.
	
	\item We consider now the view set $\V{^{CQ}}$ that contains the views as constructed in the two previous steps above.
	
	\item We find an MCR $R^{CQ}_{MCR}$ for the Datalog query $Q^{Datalog}$ using the views
	in $\V{^{CQ}}$. For building the MCR we use the inverse rule algorithm  \cite{Duschka97-3}.
	
	\item To obtain an MCR $R_{MCR}$ for $Q'$, we replace in the found MCR $R^{CQ}_{MCR}$, each $v_i^{CQ}$ by
	$v_i$, each $u_{\theta c}(X)$ by arithmetic comparison $X \theta c$ and each $u(X,Y)$ by arithmetic comparison $X\leq Y$.
\end{enumerate}

\begin{example}
	\label{ex-prime-old1}
In this example,
the reverse rule algorithm produces an MCR without including the auxiliary views, hence, we have not written these views, in order to keep things simple.
	Consider the query $Q_1$ and the views:
	\begin{center}
		\begin{tabular}{ll}
			$Q_1()$    & $\mathrm{:-}~e(X,Z),e(Z,Y),X\geq 5,Y\leq 8.$\\
			$V_1(X,Y)$ & $\mathrm{:-}~e(X,Z),e(Z,Y),Z\geq 5.$\\
			$V_2(X,Y)$ & $\mathrm{:-}~e(X,Z),e(Z,Y),Z\leq 8.$\\
			$V_3(X,Y)$ & $\mathrm{:-}~e(X,Z_1),e(Z_1,Z_2),e(Z_2,Z_3),e(Z_3,Y).$\\
		\end{tabular}
	\end{center}
	
	We have already built the Datalog program $Q_1^{Datalog}$ in Example~\ref{ex-running-dat} in a more general setting, where we assume that the query is not Boolean. Here, we use the same $Q_1^{Datalog}$ only that we delete $\overline{W}$ from all the rules. We need to add the link rules which will be with respect to constants 5 and 8 (these are the only constants that appear in the definitions).
	
	The views that will be used to apply the inverse-rule algorithm are:

	\begin{center}
		\begin{tabular}{ll}
			$V'_1(X,Y)$ & $\mathrm{:-}~e(X,Z),e(Z,Y),U_{\geq 5}(Z).$\\
			$V'_2(X,Y)$ & $\mathrm{:-}~e(X,Z),e(Z,Y),U_{\leq 8}(Z).$\\
			$V'_3(X,Y)$ & $\mathrm{:-}~e(X,Z_1),e(Z_1,Z_2),e(Z_2,Z_3),e(Z_3,Y).$\\
		\end{tabular}
	\end{center}
	Notice that we conveniently did not add any auxiliary views here because we guessed that they will not be needed.
	
	In this example, it is relatively easy to anticipate the result of applying the inverse-rule algorithm, by observing the simple form of the expansions of $Q_1^{Datalog}$. Each expansion of $Q_1^{Datalog}$ is a simple path with two unary predicates, one at one end of the path and the other at the other end. Thus, the output of the  inverse-rule algorithm is the following program. It is an MCR of $Q_1^{Datalog}$ using the views $V'_1(X,Y)$, $V'_2(X,Y)$, and  $V'_3(X,Y)$.
	
	\begin{center}
		\begin{tabular}{ll}
			$R'() $ & $\mathrm{:-}~v'_1(X,W),T(W,Z),v'_2(Z,Y).$\\
			$T(W,W)$ & $\mathrm{:-}~.$\\
			$T(W,Z)$ & $\mathrm{:-}~T(W,U),v'_3(U,Z).$\\
		\end{tabular}
	\end{center}

	The following is an MCR of the input query $Q_1$ (rather than of $Q_1^{Datalog}$) using the views  $V_1(X,Y)$, $V_2(X,Y)$ and  $V_3(X,Y)$:
	\begin{center}
		\begin{tabular}{ll}
			$R() $ & $\mathrm{:-}~V_1(X,W),T(W,Z),V_2(Z,Y).$\\
			$T(W,W)$ & $\mathrm{:-}~.$\\
			$T(W,Z)$ & $\mathrm{:-}~T(W,U),V_3(U,Z).$\\
		\end{tabular}
	\end{center}
\end{example}

\subsection{Proof that the algorithm {\bf MCR-RSI1} is correct}
The proof of the following proposition is a straightforward consequence of the construction of $R_{MCR}$ from  $R^{CQ}_{MCR}$.

\begin{proposition}
	\label{pro-conncect}
	Consider the Datalog programs $R^{CQ}_{MCR}$ and $R_{MCR}$. For each CQAC Datalog-expansion $E$ of $R_{MCR}$, there is   a CQ Datalog-expansion $E^{CQ}$ of $R^{CQ}_{MCR}$ (and vice versa), where the following is true:  The relational subgoals of $E$ are isomorphic to the purely relational subgoals of $E^{CQ}$ (by ``purely relational subgoals we mean those that do not encode ACs) and each AC in $E$ corresponds to a subgoal in $E^{CQ}$  that encodes this AC.
\end{proposition}

\begin{theorem}
	\label{thm-mcr-cont}
Given a query $Q$ which is RSI1 and views $\V$ which are
	CQAC$^{-}$s, the following is true:
	Let $R$ be a  CQAC contained rewriting of $Q$ in terms of $\V$. Then $R$ is contained in the one found by the algorithm in Subsection~\ref{subsec-buildingMCRs} Datalog$^{AC}$ program $R_{MCR}$.
\end{theorem}
\begin{proof}
	Let $R$ be a CQAC contained rewriting of $Q$ using $\V$ and let $R_{exp}$ be the view-expansion of $R$.
	We assume $R$ is AC-rectified (see Subsection \ref{subsec-rectified}). Let $Q^{Datalog}$ be the transformed query of $Q$ as in Subsection  \ref{subsec-construct-Datalog}. We argue using the following rewritings and their expansions:
	

	
	\begin{itemize}
		\item $R$ is a CQAC query which is a contained rewriting of $Q$ using $\V$.
		
		\item $R_{exp}$ is the view-expansion of $R$ with respect to $\V$.
		
		\item $R'$ is $R$ with ACs in the closure of ACs in $R$ now being relational  predicates.
		
		\item $R'_{exp}$  is the view-expansion of $R'$ with respect to $\V^{CQ}$.
		\end{itemize}
We also consider:
		\begin{itemize}

		\item $R_{exp}^{CQ}$  is $R_{exp}$  transformed into a CQ as in Subsection \ref{subsec-contained-trans}.
		
		\item $R^{CQ}$ is the rewriting that we prove can be created from $R_{exp}^{CQ}$ (we mean contained rewriting of $Q^{Datalog}$ using  $\V^{CQ}$).
		
	\end{itemize}

First, we observe that because we have the auxiliary views in $\V^{CQ}$, $R'$ is a rewriting in terms of $\V^{CQ}$ (we do not know yet whether it is a contained rewriting to the Datalog query).

The closure of ACs in $R_{exp}$ may contain: a) ACs carried over from the views definitions, b) ACs that involve only view-head variables and c) ACs that involve view-nondistinguished variables and are not in class (a) or (b).
Because of the constraint on the views to be only CQAC$^-$, the third class (c) does not exist.
The reason is that, in this class, belong ACs that are implied from at least two ACs, each one carried over from views definitions of two different view atoms in $R$. For this to happen, these two ACs should, each, relate a nondistinguished variable in a view definition with a head variable in the view definition (this is represented as view-head variable in $R_{exp}$).

Because $R$ is AC-rectified the second class (b) of ACs appear in $R$ too. Thus $R'_{exp}$ can be viewed as resulting from $R_{exp}$ by a CQ-transformation (i.e., according to the
Subsection \ref{subsec-contained-trans}). Thus we argue as follows:
Since $R$ is a contained rewriting to $Q$, $R_{exp}$ is a contained query to $Q$, and according to the results of Section 5, $R'_{exp}$ is a contained query to $Q^{Datalog}$, hence $R'$ is a contained rewriting to $Q^{Datalog}$. Hence $R'$ is contained in $R^{CQ}_{MCR}$.

Now, we  use  Proposition \ref{pro-conncect}.
$R$ and $R'$ differ only in that the ACs of one are AC-predicates of the other, in one to one fashion. For any Datalog-expansion of $R_{MCR}$ there is a Datalog-expansion of  $R^{CQ}_{MCR}$ (and vice versa) that differ in the same way. Hence the Datalog-expansion of $R^{CQ}_{MCR}$ that proves $R'$ is in
	$R^{CQ}_{MCR}$ can be used to derive a Datalog-expansion of $R_{MCR}$ that proves $R$ is in
	$R_{MCR}$.

	\end{proof}

\begin{theorem}
\label{thm-mcr-in-query}
Given a query $Q$ which is CQAC-RSI1 and views $\V$ which are
	CQAC$^{-}$s, the following is true:
	The found by the algorithm in Subsection~\ref{subsec-buildingMCRs} Datalog$^{AC}$ program, $R_{MCR}$,   is a contained rewriting.
\end{theorem}
%

\begin{proof}
	Consider a CQAC Datalog-expansion, $R$, of the found Datalog$^{AC}$ program, $R_{MCR}$. Take the view-expansion, $R_{exp}$, of $R$. Transform $R_{exp}$ into a CQ, $R_{exp}^{CQ}$,  using the construction in Subsection \ref{subsec-contained-trans}. Now, we argue in the same way as we argued in the proof of Thorem 	\ref{thm-mcr-cont} to prove that $R_{exp}^{CQ}$ is the view-expansion of  a Datalog-expansion of $R^{CQ}_{MCR}$.


According to the reverse-rule algorithm, $R_{exp}^{CQ}$  is contained in the $Q^{Datalog}$ program, hence according to the results in Section \ref{ssec:mot}, $R_{exp}$  is contained  in the query $Q$. Consequently, $R$ is a contained rewriting of $Q$.
\end{proof}

Thus we have proved:

\begin{theorem}
	\label{thm-mainsec6}
	Given a query $Q$ which is CQAC-SI1 and views $\V$ which are
	CQAC$^{-}$s, the algorithm in Subsection~\ref{subsec-buildingMCRs} finds an MCR of $Q$ using $\V$ in the language of (possibly infinite) union of CQACs.
\end{theorem}

\subsection{Building MCR for RSI1+ query}
Now, we present the algorithm for building an MCR in the language of (possibly infinite) union of CQACs for the case of CQAC$^-$ views and queries that are RSI1+. The algorithm for building  an MCR  for query $Q$ and viewset $\V$ is the following:

{\bf Algorithm MCR-RSI1+}:

\begin{enumerate}
	\item We consider query $Q'$ which results from the given query $Q$ after we have removed the ACs that contain only head variables.

\item We apply the algorithm for building MCR for query $Q'$ and views $\V$ (from previous subsection). Let this MCR be
$R'_{MCR}$.

\item We add a new rule in $R'_{MCR}$ (and obtain $R_{MCR}$ ) to compute the query predicate $Q$ as follows:
	$$Q(\overline{W}):- Q'(\overline{W}), ac_1, ac_2, \ldots$$
	where $ac_1, ac_2, \ldots$ are the ACs that we removed in the first step of the present algorithm.
\end{enumerate}

\subsection{Proof that the algorithm {\bf MCR-RSI1+} is correct}
We consider the found by the  {\bf Algorithm MCR-RSI1+}  Datalog$^{AC}$ program, $R_{MCR}$.
Theorem	\ref{thm-mcr-cont1} below proves that every CQAC contained rewriting is contained in $R_{MCR}$ and Theorem
\ref{thm-mcr-in-query1} proves that $R_{MCR}$ is a contained rewriting.

%
%
%


\begin{theorem}
	\label{thm-mcr-cont1}
Given a query $Q$ which is RSI1+ and views $\V$ which are
	CQAC$^{-}$s, the following is true:
	Let $R$ be a  CQAC contained rewriting to $Q$ in terms of $\V$. Then $R$ is contained in the one found by the  {\bf Algorithm MCR-RSI1+}  Datalog$^{AC}$ program, $R_{MCR}$.
\end{theorem}
\begin{proof}
Let $R$ be a contained rewriting to query $Q$. Since $Q'$ contains $Q$, $R$ is a contained rewriting of $Q'$ too. Hence, according to the results Theorem \ref{thm-mcr-cont}, $R$ is contained to $R'_{MCR}$.

Since $R$ is contained to $Q$, we consider the view-expansion of $R$, let it be $R_{exp}$ and we know that this is contained in $Q$, hence the containment entailment is true. However, $Q$ is a RSI1+ query, hence we can, according to Section \ref{sec:single-mapping} break the containment entailment in two as follows:

%

	\begin{center}
	\begin{tabular}{ll}
		$\beta_2\Rightarrow \mu_1(\beta_{Q'}) \vee \cdots $&\\
		$\beta_2\Rightarrow \mu_1(\beta_{Q-head}) $ & eq. (1)\\
	\end{tabular}
\end{center}

where $\beta_2 $ is the conjunction of ACs in the closure of ACs in $R_{exp} $ and $\beta_{Q'}$ is the conjunction of ACs in $Q'$, $\beta_{Q-head}$  is the conjunction of ACs that use only head variables, and $mu_i$'s are all the mappings from $Q$ to $R_{exp} $.

Observe that in equation (1), we can replace $\beta_2$ with only those ACs in the closure of $\beta_2$ that involve head variables. Because $R$ is AC-rectified, all these ACs appear in $R$; let us denote them by $\beta_{head}$
Thus $\beta_{head}$ logically implies $ \beta_{Q-head}$.
Now, $R'_{MCR}$ and $R_{MCR}$ have the same expansions, except that the latter has the ACs in $ \beta_{Q-head}$ as well. 

Hence we have concluded that a) $R$ is contained to $R'_{MCR}$ and b) the ACs in $R$ imply the added ACs in each expansion of $R'_{MCR}$ to make an expansion of $R_{MCR}$.

Now we only need to prove the following claim:
Suppose a CQAC $Q_2$ is contained in CQAC $Q_1'$. Let $Q_1$ be $Q_1'$ with some more ACs on the head variables such that these ACs are implied by the ACs in $Q_2$. Then $Q_2$ is contained in $Q_1$.

Proof of the claim: When a tuple of $Q_2$ is computed, then the same tuple is computed for $Q'_1$. However, the constants in the tuple are such that the ACs in $Q_2$ are satisfied. Since the added ACs to make $Q_1$ are implied by the ACs in $Q_2$, those ACs are satisfied too, and hence, the tuple is computed for $Q_1$ too.
\end{proof}

\begin{theorem}
\label{thm-mcr-in-query1}
Given a query $Q$ which is CQAC-RSI1 and views $\V$ which are
	CQAC$^{-}$s, the following is true:
	The found by the algorithm in Subsection~\ref{subsec-buildingMCRs} Datalog$^{AC}$ program, $R_{MCR}$,   is a contained rewriting.
\end{theorem}
\begin{proof}


Let $E$ be a CQAC query which is a Datalog-expansion of $R_{MCR}$. Let $E'$ be the CQAC that results from
$E$ by removing the head ACs. By  Theorem \ref{thm-mcr-in-query}, $E'$ is a contained rewriting in $Q'$.
Hence if we consider the view-expansion, $E'_{exp}$, of $E'$, the containment entailment is true for $E'_{exp}$ and $Q'$.

Moreover, trivially we have $\beta_{E}\Rightarrow \beta_{Q-head}$
and using the distributive law, we derive the containment entailment that shows containment of the view-expansion
$E_{exp}$ of $E$ to $Q$.
%
\end{proof}


A straightforward consequence of the above two theorems is the following theorem which is the main result of this section
\begin{theorem}
	\label{thm-mainsec6}
	Given a query $Q$ which is CQAC-SI1+ and views $\V$ which are
	CQAC$^{-}$s, the algorithm in Subsection~\ref{subsec-buildingMCRs} finds an MCR of $Q$ using $\V$ in the language of (possibly infinite) union of CQACs which is expressed by a Datalog$^{AC}$ query. 
\end{theorem}

A straightforward consequence of the above theorem and the main result in Section \ref{sec-6} is the following theorem:

\begin{theorem}
	Given a query $Q$ which is CQAC-SI1+ and views $\V$ which are
	CQAC$^{\;-}$s,  we can find all certain answers of $Q$ using $\V$ on a given view instance $\I$ in time polynomial on the size of $\I$.
\end{theorem}

\subsection{Another example}

\begin{example}
	\label{ex-prime-new}
	
	This is similar to Example \ref	{ex-prime-old1}
with slight alterations to make the point that we need  auxiliary views. The alterations are as follows: We have added a new relational subgoal $a(U)$ and a new AC on the variable of this relational subgoal in the query and we have added a relational subgoal on the same predicate on the first view. Again auxiliary views that are not used in building the MCR and are not written here.
	
	Thus, we consider the query $Q_1$ and the views:
	\begin{center}
		\begin{tabular}{ll}
			$Q_1()$    & $\mathrm{:-}~e(X,Z),e(Z,Y),X\geq 5,Y\leq 8,a(U),U\geq 46.$\\
			$V_1(X,Y)$ & $\mathrm{:-}~e(X,Z),e(Z,Y),Z\geq 5,a(Y).$\\
			$V_2(X,Y)$ & $\mathrm{:-}~e(X,Z),e(Z,Y),Z\leq 8.$\\
			$V_3(X,Y)$ & $\mathrm{:-}~e(X,Z_1),e(Z_1,Z_2),e(Z_2,Z_3),e(Z_3,Y).$\\
		\end{tabular}
	\end{center}
	
	The Datalog program $Q_1^{Datalog}$ is exactly the same as the one  in Example~\ref{ex-prime-old1} with the only alternation that the body of each mapping rule includes $U_{\geq 46}(U)$.
	
	The views that will be used to apply the inverse-rule algorithm are (now we have added one auxiliary view which we guessed will be needed):

	\begin{center}
		\begin{tabular}{ll}
			$V'_1(X,Y)$ & $\mathrm{:-}~e(X,Z),e(Z,Y),U_{\geq 5}(Z),a(Y),U_{\geq 46}(Y).$\\
			$V'_2(X,Y)$ & $\mathrm{:-}~e(X,Z),e(Z,Y),U_{\leq 8}(Z).$\\
			$V'_3(X,Y)$ & $\mathrm{:-}~e(X,Z_1),e(Z_1,Z_2),e(Z_2,Z_3),e(Z_3,Y).$\\
			$V'_4(Y)$ & $\mathrm{:-}~U_{\geq 46}(Y).$\\
			
		\end{tabular}
	\end{center}
	The new view is view $V_4'$.
	
	The Datalog program now for the query $Q$ is\footnote{Some link rules are ommitted since they are not used by the inverse rule algorithm to produce an MCR. For the same reason some coupling rules are ommitted.}:

	\begin{center}
		\begin{tabular}{l l l}
			$Q_1^{Datalog}()$       & $\symif~e(X,Y),e(Y,Z),a(U),U_{\geq 46}(U).$&\\
			&$I_{\geq 5}(X),I_{\leq 8}(Z).$ & (query rule)\\
			$J_{\leq 8}(Z)$ & $\symif~e(X,Y),e(Y,Z),a(\overline{W},Y),$&\\&$I_{\geq 5}(X).$           & (mapping rule)\\
			$J_{\geq 5}(X)$ & $\symif~e(X,Y),e(Y,Z),a(\overline{W},Y),$&\\&$I_{\leq 8}(Z).$           & (mapping rule)\\
			$I_{\leq 8}(X)$ & $\symif~J_{\geq 5}(X).$                         & (coupling rule)\\
			$I_{\geq 5}(X)$ & $\symif~J_{\leq 8}(X).$                         & (coupling rule)\\
			$I_{\geq 5}(X)$ & $\symif~U_{\geq 5}(X).$                         & (link rule)\\
			$I_{\leq 8}(X)$ & $\symif~U_{\leq 8}(X).$                         & (link rule)\\
		\end{tabular}
	\end{center}
	This is the output of the inverse rule algorithm:
	\begin{center}
		\begin{tabular}{ll}
			$R'() $ & $\mathrm{:-}~v'_1(X,W),T(W,Z),v'_2(Z,Y),v'_4(W).$\\
			$T(W,W)$ & $\mathrm{:-}~.$\\
			$T(W,Z)$ & $\mathrm{:-}~T(W,U),v'_3(U,Z).$\\
		\end{tabular}
	\end{center}

	The following is an MCR of the input query $Q_1$ using the views  $V_1(X,Y)$, $V_2(X,Y)$ and  $V_3(X,Y)$. Notice that we replace view $v'_4(W)$ by $W\geq 46$.
	\begin{center}
		\begin{tabular}{ll}
			$R() $ & $\mathrm{:-}~V_1(X,W),T(W,Z),V_2(Z,Y),W\geq 46.$\\
			$T(W,W)$ & $\mathrm{:-}~.$\\
			$T(W,Z)$ & $\mathrm{:-}~T(W,U),V_3(U,Z).$\\
		\end{tabular}
	\end{center}
\end{example}

\subsection{Extending the result}
The following example shows that the result of Theorem~\ref{thm-mcr-cont} (hence the result of Theorem \ref{thm-mainsec6})  cannot be extended to include views that are in the language of CQACs.

\begin{example}
	Suppose we have the following query and views:
		\begin{center}
		\begin{tabular}{l}
			$Q$:   $q(Y):- a(Y,X),b(X,Z),X\geq 5, Z\leq 6$\\
			$V_1$:  $v_1(Y):- a(Y,X),b(X,X'),X\geq Y, X'\leq 6$\\
			$V_2$:  $v_2(Y):- a(Y,Z'),b(Z',Z),Y\geq Z, Z'\geq 5$\\
		\end{tabular}
	\end{center}
%
%
%
	
Here, the view definitions violate the constraint that no AC should be included between a head variable and a nondistinguished variable.

	The following is a contained rewriting, for which we will argue that the technique in this section does not work:
			
	\begin{center}
		\begin{tabular}{l}
			$R$:  $q(Y):-v_1(Y),v_2(Y)$\\
		\end{tabular}
	\end{center}
	
	The view-expansion of $R$ is:
	
		\begin{center}
		\begin{tabular}{l}
			$R_{exp}$: $q(Y):- a(Y,X),b(X,X'), X'\leq 6, a(Y,Z'),b(Z',Z), Z'\geq 5$\\
		\end{tabular}
	\end{center}
	
	and after we transform it to a CQ  it becomes:
	
		\begin{center}
		\begin{tabular}{l}
			$R^{CQ}_{exp}$: $q(Y):- a(Y,X),b(X,X'), U_{\leq 6}(X'), a(Y,Z'),b(Z',Z), U_{\geq 5}(Z')$\\
		\end{tabular}
	\end{center}
	
	
	Suppose we transform $Q$ to $Q^{Datalog}$, then the following is the only  Datalog-expansion of the program that could be used to prove that $Q^{Datalog}$ contains $R^{CQ}_{exp}$.

			\begin{center}
		\begin{tabular}{l}
			$E$: $a(Y,X),b(X,X'), U_{\leq 6}(X'), a(Y,Z'),b(Z',Z), U_{\geq 5}(Z'),U(Z,X)$\\
		\end{tabular}
	\end{center}
	
	We see that the body of $R^{CQ}_{exp}$ and $E$ will be isomorphic if we append $U(Z,X)$ to $R^{CQ}_{exp}$, which is equivalent to appending $Z\leq X$ to $R_{exp}$, which, further means appending $Z\leq X$ to $R$. This is not possible however, since in $R$ there is only one variable. This remark highlights the reason for  failing to extend the result beyond  Theorem \ref{thm-mainsec6} and probably, in future work, complexity results will prove that it is rather  impossible to be extended.
	
	%
	%
	%
	%
\end{example}


\section{Conclusions}
In this paper we have investigated into the computational complexity of query containment for CQACs and of computing certain answers in the framework of answering CQAC queries using CQAC views. Our results point to several directions for future research. 

Candidates for which  the problem of CQAC query containment may be  $\Pi^p_2$-complete are the following:
 a) The containing query  contains only LSI  (both open and closed). b) The containing query contains two closed LSI and two closed RSI. Probably it is still $\Pi^p_2$-complete if the contained query is restricted, e.g., with only SI arithmetic comparisons. 

We know that there are classes of queries where the CQ query containment problem is polynomial e.g., when the containing query is acyclic.
We conjecture that, for many such classes of queries, if we consider the classes of queries as in Section 5 but with the containing query having relational subgoals as the CQ queries that have a polynomial time algorithm, then
the containment problem is polynomial. E.g., if the relational subgoals of the containing query form an acyclic hypergraph and there are only several closed LSIs and one closed RSI on the nondistinguished variables, then 
testing containment may be done in polynomial time.


In the framework of answering queries using views, besides using any new results for query containment to extend the results in Section 7 further, we have started a discussion in  \ref{sec-app-unclean}  for CQs  to see whether MCRs can be useful for handling certain types of unclean data.
This discussion may be extended to include   CQACs.

\bibliographystyle{abbrv}
\bibliography{references}
\appendix
\section{Proof of Theorem \ref{thm:cont-CQAC}}
\label{prf-thm:cont-CQAC}
\begin{proof}
	One of the directions is straightforward: If the containment entailment is true, then in any database that satisfies $\beta'_2$, one of the $\mu_i(\beta' _1)$ will be satisfied (because we deal with constants), and hence containment is proven.

	For the ``only-if'' direction,
	suppose $Q_2$ is contained in $Q_1$, but the containment entailment is false.
	We assign constants to the variables that make this implication false. Then for all the containment mappings $\mu _i$ 
 the query containment is false, because
	we have found a counterexample database $D$.
	Database $D$ is constructed by assigning the corresponding constants to the ordinary subgoals of $Q_2$.
	On this counterexample database $D$, $Q_2$ produces a tuple,
	but there is no $\mu_i$ that will make $Q_1$ produce the same tuple (because all $\mu_i(\beta'_1)$
	fail). We need to remember that, using the $\mu_i$'s, we can produce {\sl all} homomorphisms from
	$Q_1$ to any database where the relational atoms of $Q_2$ map via a homomorphism. This is because
	the $\mu_i$'s were produced using the normalized version of the queries -- and, hence, $\mu_i$'s were not constrained by duplication of variables or by constants (recall that, in a homomorphism, a variable is allowed to map to a single target and a constant is allowed to map on the same constant).
\end{proof}


\section{Proof of Theorem \ref{thm:main123}}
\label{prf:thm-main123p}
\begin{proof}
	We consider the canonical database, $D$, of $Q_2^{CQ}.$ For convenience, the constants in the canonical database use the lower case letters of the variables they represent. Thus, constant $x$ is used in the canonical database to represent the variable $X$. We will use the containment test that says that a
	Datalog query contains a conjunctive query $Q$  if and only if the Datalog query computes the head of $Q$  when applied on the canonical database of the conjunctive query $Q$.
	
	%
	%
	

	
	{\bf ``If'' direction:}
{\it Remark:}
By construction of the transformed queries $Q_1$ and $Q_2$  into a Datalog query and a CQ query respectively, the following happens: The computation (i.e., the firing of the mapping rules during this computation) of the head of $Q_2^{CQ}$  when the Datalog query is applied on the canonical database of  $Q_2^{CQ}$  gives a set of 
mappings from $Q_{10}$ to $Q_{20}$ that are exactly the mappings that make the containment entailment true. Hence this direction holds even if the containing query has any number of LSIs and RSIs in it. Therefore, it provides an incomplete test for the case  the containing query has any number of LSIs and RSIs: If $Q_2^{CQ}$ is contained in the Datalog query, then $Q_2$ is contained in $Q_1$, otherwise we do not know.

	The proof of this direction is done by induction on the number of the times some mapping rule is fired during the computation of a $J$ fact in a computation that uses the shortest derivation tree.

	{\bf Inductive Hypothesis:} If, in the computation of a $J$ fact associated with AC $e$, we have used the
	mappings $\mu_1, \mu_2, \ldots, \mu_k$ (via applications of mapping rules), where $k<n$ then the following holds:
	$$\beta_2\Rightarrow \mu_1(\beta_1) \vee \mu_2(\beta_1) \vee \cdots \vee \mu_k(\beta_1)\vee \neg e$$
	
	{\sl Proof of Inductive Hypothesis. }
	
	The base case is straightforward, since it is the case when a $J$ fact is computed after the application of one mapping rule, say by mapping $\mu_i$. This is enabled because each of the ACs in the $\mu_i(\beta_1)$ except one (the one associated with the computed $J$ fact) are directly implied by $\beta_2$.
	
	Suppose we compute a fact via $n$ mappings. Then all its $I$ facts used in the computation are computed via at most $n-1$ mappings, hence the inductive hypothesis holds for the corresponding $J$ facts that were used to compute the $I$ fact via a coupling rule.

	According to the inductive hypothesis,
	each such $I$ fact that is computed via a $J$ fact, which in turn was computed via
	some mappings $\mu_{ij}, j=1,2,\ldots$ (i.e., these are the mappings for all mapping rules that were applied during the whole computation of $J$ fact), implies that  the following is true:
	$$\beta_2\Rightarrow \mu_{i1}(\beta_1) \vee \mu_{i2}(\beta_1) \vee \cdots \vee \mu_{il_i}(\beta_1)\vee \neg e_i$$
	or equivalently:
	$$\beta_2 \wedge \neg \mu_{i1}(\beta_1) \wedge \neg \mu_{i2}(\beta_1) \wedge \cdots \wedge \neg \mu_{il_i}(\beta_1)  \Rightarrow \neg e_i$$
	Thus,   we can combine the above implications for all $I$ facts used for the current application of a mapping rule and have that the following is true:
	
	$$\beta_2\wedge  \bigwedge_{for~~all~~ i}[\neg \mu_{i1}(\beta_1) \wedge \neg \mu_{i2}(\beta_1) \wedge \cdots \wedge \neg \mu_{il_i}(\beta_1)]\Rightarrow  \neg e_1 \wedge \cdots~$$
	
	We write the above in the  form:
	$$\beta_2 \wedge \neg \mu_1(\beta_1) \wedge \neg \mu_2(\beta_1) \wedge \cdots \Rightarrow \neg e_1 \wedge \cdots~~~~~~~~~~(2)$$
	where for simplicity we have expressed the  $\mu_{ij}, i=1,2,\ldots,  j=1,2,\ldots$ as $\mu_1, \mu_2, \ldots$.
	Now $e_1, \dots$ are the ACs each associated with the $J$ facts used for this mapping rule.

	Suppose we apply mapping rule via mapping $\mu_{current}$ that uses $I$ facts computed in previous rounds using at most $n-1$ mappings.

	
	When a  coupling rule of the first kind is fired then the two variables in the rules are such that their associated ACs $e_I$ and $e_J$  are such that
	$e_I \vee e_J$ is true.
	
	When a  coupling rule of the second kind is fired then the two  variables $X,Y$ of the rule (which appear also in the binary EDB $U(X,Y)$ in the body of the rule)
	are  instantiated to constants in the canonical database of $Q_2$ whose
	corresponding variables $X',Y'$ in $Q_2$  are such that $X'\leq Y'$. Hence we have for the  associated ACs $e_I$ and $e_J$ of the IDB predicates of the rule that $\beta_2 \Rightarrow e_I \vee e_J$.
	
	Thus, in any case, when a coupling rule is fired, the following is true: $\beta_2 \Rightarrow e_I \vee e_J$. We use this remark in the second implication (second arrow)
	of implications \ref{eq-34} below.
	%


	From (2), we have for each $e_i$:
	
	$$\beta_2 \wedge \neg \mu_1(\beta_1) \wedge \neg \mu_2(\beta_1) \wedge \cdots \Rightarrow \neg e_i$$ which yields:
	\begin{equation}\label{eq-34}
	\beta_2 \wedge \neg \mu_1(\beta_1) \wedge \neg \mu_2(\beta_1) \wedge \cdots \Rightarrow \neg e_i\wedge \beta_2\Rightarrow ac_i
	\end{equation}

	Hence we have, combining all ACs in $\mu_{current}(\beta_1)$,
	
	$$\beta_2 \wedge \neg \mu_1(\beta_1) \wedge \neg \mu_2(\beta_1) \wedge \cdots \Rightarrow ac_1\wedge ac_2\cdots \wedge \neg e_t$$
	where $e_t$ is one of the $ac_i$'s and  is the associated AC to the J facts computed by the mapping $\mu_{current}$ which is used to fire a mapping rule.
	
	%
	Thus,
	we get
	$$\beta_2 \wedge \neg \mu_1(\beta_1) \wedge \neg \mu_2(\beta_1) \wedge \cdots \Rightarrow \mu_{current}(\beta_1)\vee \neg e_{current}$$
	from which we get the implication in the inductive hypothesis.
	
	To finish the proof of this direction, we need to argue about the application of the query rule via mapping $\mu$, whose only difference with a mapping rule is that all the ACs in the $\mu(\beta_1)$ are coupled, hence we derive finally the containment entailment.

	{\bf  ``Only-if'' direction}:
	Assume the containment entailments holds, i.e., the following holds:
	$$\beta_2  \Rightarrow \mu_1(\beta_1) \vee \cdots \vee \mu_l(\beta_1)$$
	We will prove this direction by induction on the number of $\mu_i$s in a  containment entailment that is true and uses the minimal number of mappings.
	
	{\bf Inductive hypothesis }
	For $k\leq n$, there is a set of mappings among the ones in the containment entailment, i.e., let them be $\mu_{m+1}, \mu _{m+2}, \ldots, \mu_l$
	(where $k=l-m$) such that the following two happen:

	(i) The following is true:
	$$\beta_2  \Rightarrow \mu_1(\beta_1) \vee \cdots \vee \mu_m(\beta_1) \vee $$
	$$\mu _{m+1}(e^{\beta_1}_{m+1}) \vee \mu _{m+2}(e^{\beta_1}_{m+2}) \vee  \cdots   $$
	where $m<l$ and $e^{\beta_1}_{m+1}, e^{\beta_1}_{m+2} , \dots$ are ACs from $\beta_1$ such that  each
$\mu _{i}(e^{\beta_1}_{i})$ is not directly implied by $\beta_2$.
	
	(ii)  The  mappings $\mu_{m+1},
	\mu_{m+2}, \ldots$  are used
	to compute the facts $I(x_{m+1}), I(x_{m+2})\ldots$
	where $x_{m+1}$ (similarly for $x_{m+2}$, etc ) represents the variable
	in $\mu _{m+1}(e^{\beta_1}_{m+1})$.  More specifically when we compute the
	fact $I_{\theta c}(x_{m+1})$ then $\mu _{m+1}(e^{\beta_1}_{m+1})$ is the AC $X_{m+1}\theta c$.
	The same for $I(x_{m+2})$, etc.
	
	Proof of the inductive hypothesis. The base case is when $k=l-m=0$ and is a consequence of 
Proposition
\ref{trick-pro}.

To prove for $k=n+1$, we argue as follows. We begin with the implication in the inductive hypothesis, and use
	Proposition \ref{trick-pro1} . According to this proposition there is a mapping, let it be
	$\mu_m$ such that the following is true:
	$$\beta_2  \Rightarrow \mu_1(\beta_1) \vee \cdots \vee \mu_{m-1}(\beta_1) \vee $$
	$$\mu_m(e^{\beta_1}_{m}) \vee\mu _{m+1}(e^{\beta_1}_{m+1}) \vee \mu _{m+2}(e^{\beta_1}_{m+2}) \vee  \cdots   $$

	This proves part (i) of the inductive hypothesis. Now we need to prove the (ii) of the inductive hypothesis, i.e., that
	we can compute the fact asscociated with $\mu_m(e^{\beta_1}_{m})$.
	
	

	According to (b) in Proposition \ref{trick-pro1} and taking into account Lemma 	\ref{lemma-LSI-appendix-NP},
	we have that for each $ac_i$ in
	$\mu_m(\beta_1)$ the following is true: $$\beta_2\Rightarrow ac_i \vee e_j$$
	These (i.e., the corresponding variables of the constants in $ac_i$ and $e_j$) provide the instantiation for the firing of couplings rules that compute all the I facts necessary to fire a mapping rule by the instantiation provided by $\mu_m$.
	\end{proof}
	%
	%
	
	%
	%
	%

\section{Semi-monadic Datalog - Containment}
\label{sec:semi-monadic-datalog-cont}
Here, we prove a stronger result than the one we need to prove Theorem \ref{thm:datalog-np-complete}.
We define semi-monadic Datalog and prove that the problem of containment of a conjunctive query to a semi-monadic Datalog is in NP.

\subsection{Definition of Semi-monadic Datalog}

A {\em binding pattern} is a vector consisting of b (for bound) and f (for free) and its length is defined to be equal to the number of components in the vector, e.g., $bbfb$ is a binding pattern of length 4. An annotated predicate is a predicate atom together with a list of its variables and a binding pattern of length equal to the length of the list of the variables. E.g., $P^{bbf}_{X,Y,Z}(X, Y,X, Z)$ is an
annotated predicate; the list of the variables has been put as subscript in the name of the predicate. Notice that $P^{bbf}_{X,Y,Z}(X, Z,X, Y)$ is a different annotated predicate. However
$P^{bbf}_{X,W,Z}(X, W,X, Z)$ is the same as  $P^{bbf}_{X,Y,Z}(X, Y,X, Z)$ because names of variables do not matter.

Now, in order to prove that a Datalog query is semi-monadic, we do as follows: We start with the query predicate and annotate it with the binding pattern with all b's.
For each already annotated IDB predicate and for each rule with head this predicate, we unify the arguments in the annotated predicate with the head of the rule. Those variables that are bound in the pattern are also bound in the IDB predicates in the body of the rule, all other variables are free. Thus, for each IDB predicate in the body of the rule, we create a new annotated IDB predicate by providing the binding pattern for its variables.

The above procedure will stop because there is only a finite number of distinct annotated predicates. If each annotated predicate constructed has a binding pattern with only one f, then
we say that the Datalog query is {\em semi-monadic}.

Remark:
A computation of a fact  $F$ for a semi-monadic Datalog query will use in the derivation tree only IDB facts that have the property: All the variables in the fact have values that are one of the constants of the fact $F$, except one, which may have, in general, any value. To show this remark, consider, towards contradiction, that there is a fact in the derivation tree for which this is not true. Then, by considering the path, in the derivation tree, from its root to this fact, it is easy to show that this succession of rules would have created an annotated IDB predicate with
more than one f's.

In the case of Theorem \ref{thm:datalog-np-complete}, all IDB predicates have binding patterns $b\ldots bf$ of the same length.


\subsection{Containment test. }
\label{prf:np-complete-non-monadic}

\begin{theorem}
	\label{thm-semi-monadic-np}
	Consider a pair ($Q_1$, $Q_2$) where $Q_1$ is a semi-monadic Datalog query and $Q_2$ is a CQ.
	Then testing containment of $Q_2$ to $Q_1$ is in NP.
\end{theorem}
\begin{proof}
We claim that, during a computation of the Datalog query on an input, only annotated IDB predicates are populated with  facts (that are computed during the computation) that have in their pattern at most one f. 
Towards contradiction, suppose there is a computation where an annotated IDB predicate fact appears with more than one f on its pattern.  Take a path from the root of the derivation tree to this IDB fact. This path tells you that there is a sequence of rules, that, if taken, during the process of creating annotating IDBs, you will arrive in this annotated IDB. We need to prove that the length of this path is bounded.  This  is easy because the annotations are finite and if  the path is  long then  annotations will appear more than once.

After the above observation, we assume wlog that the Datalog query is monadic.

Consider a shortest derivation tree, $T_F$, of the fact $F$. A shortest derivation tree is one with the shortest height, where the height of a tree is defined to be the length of the longest path from the root to a leaf. We define the {\em level} of node $u$ in $T_F$ to be the height of the subtree that is rooted in $u$.

\begin{proposition}
	\label{prop:derivation-dag}
	Consider a Datalog query $P$ and a derivation tree $T_F$ for a fact $F$. Then, all identical facts in $T_F$ are at the same level.
\end{proposition}

\begin{proof} If not then consider the fact  residing on the node with the smallest level and replace all subtrees with the subtree rooted in this node. This does not increase the level of a node.
\end{proof}

After these observations, we continue with the proof of the present theorem.

	We will prove that the following decision problem is in NP:
	We consider the canonical database, D,  of the CQ $Q_2$ and we compute the Datalog query $Q_1$ on $D$ and derive the output $Q_1(D)$.
	Given a fact $F$, we ask the question whether $F$ is  in $Q_1(D)$.

	%
	
	Let us now consider that the following certificate is given:
	\begin{itemize}
		\item The unary IDB facts computed (which are polynomially many).
		\item The derivation DAG $G$ that computes the IDB facts (polynomial in size).
		
	\end{itemize}
	
	Following the Proposition~\ref{prop:derivation-dag}, we prove that the following construction results in a Directed Acyclic Graph (DAG). We consider a derivation tree $T_F$ for fact $F$. We collapse all subsets of identical facts into a single node (the edges of the tree are retained). We call this a {\em derivation DAG} of the fact $F$ and denote $G_F$.
	
	We have not proved yet that indeed $G_F$ is a DAG. Notice that each path on $G_F$ uses the same edges as the edges in $T_F$ and, thus, it corresponds to a path in $T_F$.
However, an edge in  $T_F$ joins two nodes in different levels. Hence a cycle in $G_F$, which corresponds to a path  in $T_F$  cannot exist because a path in $T_F$ joins nodes in descending levels. 
	
	The derivation DAG $G$ that is considered above is depicted in Figure~\ref{fig:derivationtreetheoremappendix}.
	Each node in the $G$ is either a EDB fact in the canonical database of $Q_2$
	or a tuple consisting of the following:
	\begin{itemize}
		\item the rule in $Q_1$ query that is fired to compute the IDB fact from the previously computed IDB facts and/or EDB facts,
		\item a topological order of the graph $G$, and
		\item the mapping from the variables of the rule to the constants in $D$ producing the IDB fact.
	\end{itemize}
	Each directed edge $(n_1,n_2)$ of $G$ describes that the fact $n_1$ is used to compute the fact $n_2$.
	
	\begin{figure}
		\centering
		\includegraphics[width=0.7\linewidth]{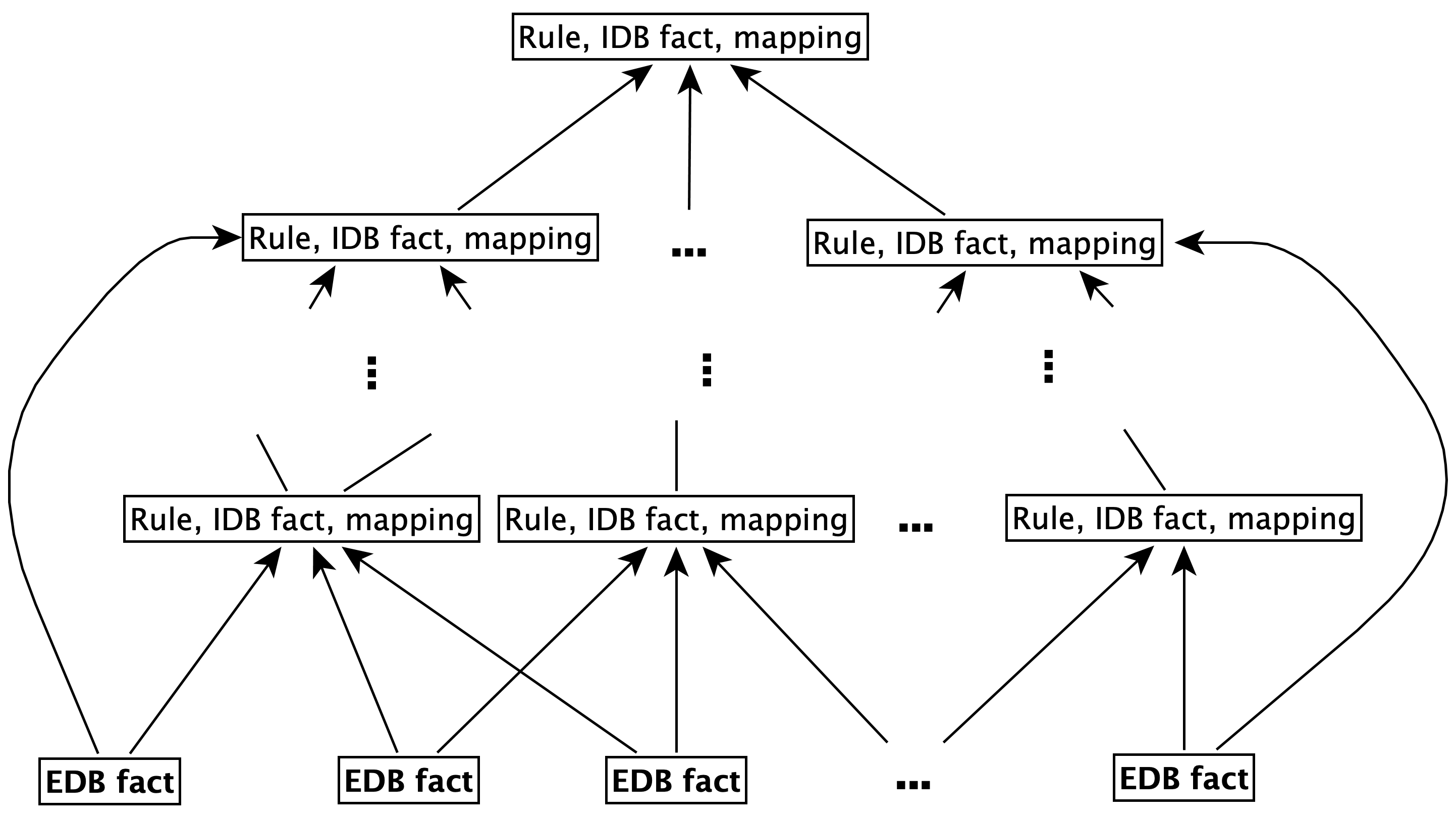}
		\caption{Derivation DAG}
		\label{fig:derivationtreetheoremappendix}
	\end{figure}
	
	Considering now such a certificate. To test it, we perform the following:
	\begin{enumerate}
		\item we check whether the given graph is a DAG following the topological order,
		\item for each non-EDB node, we apply the mapping on the rule, check whether the head of the rule equals the fact in the node and that all the facts used in the application of this rule are computed, which means they are on lower (in the topological ordering) nodes.
	\end{enumerate}
	
	It is easy to verify that the aforementioned tests can be validated in polynomial time, which proves that the problem is in $NP$.
	%
	%
	%
	%
	%
\end{proof}

\section{Constraints (tgds and egds) and the chase algorithm}
\label{pre:dependencies-chase} 

We define tuple-generating dependencies (tgds, for short) and  equality-generating dependencies (\textit{egds}, for short). Then, we describe the Chase algorithm, a significant tool for reasoning about dependencies.


\begin{definition}
	Let $\textbf{S}$ be a database schema. A {\em tuple-generating dependency} is  defined by a formula of the following form:
	$$d_t:\;\phi(\overline{X})\rightarrow\psi(\overline{X},\overline{Y}),$$
	where $\phi$ and $\psi$ are conjunctions of atoms with predicates in $\textbf{S}$, and $\overline{X}$, $\overline{Y}$ are vectors of variables. A {\em equality-generating dependencies} is  defined by a formula of the following form:
	$$d_e:\;\phi(\overline{X})\rightarrow(X_1=X_2),$$
	where $\phi$ is a conjunctions of atoms with predicates in $\textbf{S}$, and $\overline{X}$ is vector of variables and $X_1$, $X_2$ are included in $\overline{X}$.
	Considering a database instance $D$ of $\textbf{S}$, we say that $D$ {\em satisfies} $d_t$ if whenever there is a homomorphism $h$ from $\phi(\overline{X})$ to $D$, there exists an extension $h'$ of $h$ such that $h'$ is a homomorphism from $\phi(\overline{X}) \wedge \psi(\overline{X},\overline{Y})$ to $D$. In addition, we say that $D$ {\em satisfies} $d_e$ if for each homomorphism $h$ from $\phi(\overline{X})$ to $D$, we have that $h(X_1)=h(X_2)$.
\end{definition}

%
%
%

We define the \textit{chase step}, the building block of the chase algorithm:
\begin{definition}
	Let $\CS$ be a database schema  and $D$ be a database instance of $\CS$. Consider also the following dependencies:
	\begin{center}
		\begin{tabular}{ll}
			$d_t:$ & $\phi(\overline{X})\rightarrow\psi(\overline{X},\overline{Y})$,\\
			$d_e:$ & $\phi(\overline{X})\rightarrow(X_1=X_2)$, \\
		\end{tabular}
	\end{center}
	where $\phi$, $\psi$ are conjunction of atoms with predicate in $\CS$.
	Then, the {\em chase step} for the dependencies $d_t$ and $d_e$ is defined as follows.
	\begin{description}
		\item[(tgd $d_t$)] Let $h$ be a homomorphism from  $\phi(\overline{X})$ to $D$ such that there is no extension $h'$ of $h$ that maps $\phi(\overline{X})\wedge\psi(\overline{X},\overline{Y})$ to $D$. In such a case, we say that $d_t$ can be applied to $D$ and we define the database instance $D'=D\cup\CF_{\psi}$, where $\CF_{\phi}$ is the set of atoms of $\psi$ obtained by substituting each variable $x$ in $\overline{X}$ with $h(x)$ and each variable in $\overline{Y}$ (not mapped through $h$) with a fresh variable; i.e.,  $\CF_{\phi}=\{\psi(h(\overline{X}),\overline{Y})\}$. The fresh variables used to replace variables in $\overline{Y}$ are called {\em labeled nulls}. We say that the result of applying $d_t$ to $D$ with $h$ is $D'$ and write $D\xrightarrow{d_t,h}D'$ to denote the chase step on $D$ with the tgd $d_t$. 
		\item[(egd $d_e$)] Let $h$ be a homomorphism from  $\phi(\overline{X})$ to $D$ such that $h(X_1)\neq h(X_2)$. In such a case, we say that $d_e$ can be applied to $D$ and we define the database instance $D'$ as follows:  
		\begin{itemize}
			\item if there is a fact $e$ in $\{\phi(h(\overline{X}))\}\cap D$ such that $h(X_2)$, $h(X_1)$ are constants and $h(X_2)\neq h(X_1)$ then $D'=\bot$; otherwise 
			\item for each fact $e$ in $\{\phi(h(\overline{X}))\}\cap D$, we replace $h(X_2)$ with $h(X_1)$ and add it into $D'$.
		\end{itemize} 
		We say that the result of applying $d_e$ to $D$ with $h$ is $D'$ and write $D\xrightarrow{d_e,h}D'$ to denote the chase step on $D$ with the egd $d_e$. If $D'=\bot$, we say that the step {\em fails}.
	\end{description}
\end{definition}

Then, the chase algorithm is defined as follows.

\begin{definition}
	Let $\C$ be a set of tgds and egds and $D$ be a database instance. Then, we define the following.
	\begin{itemize}
		\item A {\em chase sequence} of $D$ with $\C$ is a sequence of chase steps $D_i\xrightarrow{d_i,h_i}D_{i+1}$, where $i=0,1,\dots$, $D_0=D$ and $d_i\in\C$.
		\item A {\em finite chase} of $D$ with $\C$ is a finite chase sequence $D_i\xrightarrow{d_i,h_i}D_{i+1}$, with $i=0,1,\dots, n$, $D_0=D$ and $d\in\C$, such that either $n$-th step fails, or there is no dependency $d\in\C$ and there is no homomorphism $h$ such that $d$ can be applied to $D_n$.
	\end{itemize}
\end{definition}

The following theorem states the property of chase that makes it useful:
\begin{theorem}
	\label{thm-chase-main22}
	Let  $\C$  be a set of tgds, and $D$ a database instance that satisfies the dependencies in $\C$. Suppose $K$ is a database instance, such that there exists a homomorphism $h$ from $K$ to $D$.
	Let $K_{\C}$ be the result of a successful finite chase on $K$ with the set of dependencies $\C$. Then the homomorphism $h$ can be extended to a homomoprhism $h'$ from $K_{\C}$ to $D$.
\end{theorem}

\subsection{Checking CQAC query containment under constraints}
\label{pre:dependencies-chase1}

We now define the query containment in the presence of constraints. Considering a database schema $\CS$, a set of constraints $\C$ over $\CS$ and two CQs $Q_1$, $Q_2$ over $\CS$, we say that $Q_2$ is contained in $Q_1$ under the constraints $\C$, denoted $Q_2\sqsubseteq_{\C} Q_1$, if for all databases $D$ that satisfy $\C$ we have that
$Q_2(D)\subseteq Q_1(D)$.

Considering a CQ $Q$ over a database schema $\CS$ and a set $\C$ of tgds and egds dependencies over $\CS$, we construct the canonical database $D$ of $Q$  and apply the chase algorithm. Let $D'$ be the database resulted by chase. We can now construct a CQ $Q_{\C}$ from $Q$ and $D'$, such that the head of $Q$ equals the head of $Q_{\C}$ and $Q_{\C}$'s body is constructed by de-freezing the constants back to their corresponding variables. If the chase of $Q$ with $\C$ terminates then $Q_{\C}$ is the called the \textit{chased query} of $Q$ with $\C$. In such a case, for all databases $D$ that satisfy the constraints $\C$, we have that $Q(D)\subseteq Q_{\C}(D)$ (i.e., $Q\sqsubseteq_{\C} Q_{\C}$) \cite{BeeriV84}.

Given a set $\C$ of constraints that are tgds and egds, when a database t-instance $D'$ results after applying the chase algorithm with $\C$ on database t-instance $D$, we say that $D'$ is $D$ {\em after chased with $\C$}. If $D$ is a canonical database of a CQAC query $Q$, then we say that  $D'$ is a $\C$-canonical database of 
$Q$. 

\begin{theorem}
	A CQAC query $Q_2$ is contained into a CQAC query $Q_1$ under constraints $\C$ (denoted $Q_2\sqsubseteq _\C Q_1$) if and only if, for each database $D$ belonging to the set of $\C$-canonical databases of $Q_2$ with respect to $Q_1$, the  query $Q_1$ computes all the tuples that $Q_2$ computes if applied on $D$.
\end{theorem}

\section{Certain answers, MCRs and unclean data}
\label{sec-app-unclean}


Under the OWA, if $\I\not\subseteq\V(D)$ for all databases $D$, and $R$ is
an MCR of a query $Q$ using views with respect to a query language $\L$, then
there are cases where $\text{certain}(Q,\I)=\emptyset$ and
$R(\I)\neq\emptyset$. We illustrate on an example.

\begin{example}\label{example_technicality_cqs}

	Consider the case where the query is $Q(x,y)~\hbox{\rm :-}~\ a(x,y)$, we have
	only one view $v(x,x,y)~\hbox{\rm :-}~\ a(x,y)$, and the view instance is
	$\I=\{v(1,2,3), v(4,4,5)\}$. Since $v(1,2,3)\in\I$ and
	$v(1,2,3)\not\in\V(D)$ for any database $D$, we have that
	$\I\not\subseteq\V(D)$ for all databases $D$.
	
	There is only one rewriting $R(x,y)~\hbox{\rm :-}~v(x,x,y)$.  $R(\I)=\{(4,5)\}$, hence:
	\[\text{certain}(Q,\I)=\bigcap_{D\ s.t.\ \I\subseteq\V(D)}Q(D)=\emptyset \text{ \ because } \not\exists D \text{ such that \ } \I\subseteq\V(D).\]
\end{example}

We define, however, new semantics that can be useful in data cleaning, in the following way: We do not need to change (clean) the data and still get correct answers. Observe is our example, that the answer we got by applying the MCR to our data, $\I$, is correct in the following sense: The instance $\I$ in the example contained one tuple that could not have been produced by applying the view on any instance. Hence we can assume that this tuple is incorrect and define a new $\I '$ as follows:

\begin{itemize}
	\item  $\I '$  is maximal with respect to the property below.
	\item   $\I '$ is a subset of $\I$ and there is a $D$ such that $\I'\subseteq\V(D)$.
	
\end{itemize}

It is observed that, for CQ queries and views,  $\I'$ is unique and   we can produce $\I'$ by removing in any order facts from $\I$ until the property $\I'\subseteq\V(D)$ is satisfied.
Formally, we have:

\begin{definition}
	We define the {\em maximal consistent} view instance of $\I$ to be instance $\I$ with the property: it is the maximal subset of $I$ such that there is a database $D$ such that $I'\subseteq \V(D)$.
\end{definition}

\begin{definition}
	
	We  define the {\em correct certain answers } of $\I$ to be the certain answers of $\I'$ which is the maximal consistent view instance of $\I$.
\end{definition}
Our  conjecture is the following:

	{\bf Conjecture:}
For CQ queries and views, an MCR of the query using the views produces all correct certain answers of any view instance $\I$.

%

\end{document}